\documentclass[10pt, letterpaper]{IEEEtran}

\usepackage[colorlinks,urlcolor=blue,linkcolor=blue,citecolor=blue]{hyperref}

\usepackage{array}
\usepackage[hang]{footmisc}
\newcommand{\norm}[1]{\left\lVert#1\right\rVert}
\usepackage[subrefformat=parens,labelformat=parens,caption=false,font=footnotesize]{subfig}
\usepackage{amsmath}
\usepackage[nolist,nohyperlinks]{acronym}
\usepackage{amsmath}
\DeclareMathOperator{\argmax}{arg\,max}

\begin{acronym}
    \acro{4G}{fourth generation}
    \acro{5G}{fifth generation}
    \acro{AoA}{angle of arrival}
   	\acro{AoD}{angle of departure}
   	\acro{AP}{access point}
    \acro{BCRLB}{Bayesian CRLB}
    \acro{BG}{beam group}
    \acro{BS}{base stations}
    \acro{BPP}{binomial point process}
    \acro{BP}{broadcast probing}
	\acro{CDF}{cumulative density function}
    \acro{CF}{closed-form}
    \acro{CRLB}{Cramer-Rao lower bound}
    \acro{ECDF}{empirical cumulative distribution function}
    \acro{EI}{Exponential Integral}
    \acro{eMBB}{enhanced mobile broadband}
    \acro{FIM}{Fisher Information Matrix}
    \acro{GoF}{goodness-of-fit}
    \acro{GPS}{global positioning system}
    \acro{GE}{group exploration}
    \acro{GNSS}{global navigation satellite system}
    \acro{HetNets}{heterogeneous networks}
    \acro{IoT}{internet of things}
    \acro{IIoT}{industrial internet of things}
    \acro{LOS}{line of sight}
    \acro{LLR}{log-likelihood ratio}
    \acro{MAB}{multi-armed bandit}
    \acro{MBS}{macro base station}
    \acro{MEC}{mobile-edge computing}
    \acro{mIoT}{massive internet of things}
    \acro{MIMO}{multiple input multiple output}
    \acro{mm-wave}{millimeter wave}
    \acro{mMTC}{massive machine-type communications}
    \acro{MS}{mobile station}
    \acro{MVUE}{minimum-variance unbiased estimator}
    \acro{NLOS}{non line-of-sight}
    \acro{OFDM}{orthogonal frequency division multiplexing}
    \acro{PAC}{probably approximately correct}
    \acro{PDF}{probability density function}
    \acro{PGF}{probability generating functional}
    \acro{PLCP}{Poisson line Cox process}
    \acro{PLT}{Poisson line tessellation}
    \acro{PLP}{Poisson line process}
    \acro{PPP}{Poisson point process}
    \acro{PV}{Poisson-Voronoi}
    \acro{QoS}{quality of service}
    \acro{RAT}{radio access technique}
    \acro{RL}{reinforcement-learning}
    \acro{RSSI}{received signal-strength indicator}
    \acro{BS}{base station}
   	\acro{SINR}{signal to interference plus noise ratio}
    \acro{SNR}{signal to noise ratio}
    \acro{SSB}{synchronization signal block}
    \acro{SWIPT}{simultaneous wireless information and power transfer}
    \acro{TS}{Thompson Sampling}
    \acro{TS-CD}{TS with change-detection}
    \acro{KS}{Kolmogorov-Smirnov}
    \acro{UCB}{upper confidence bound}
	\acro{ULA}{uniform linear array}
    \acro{UPA}{uniform planar array}
	\acro{UE}{user equipment}
 	\acro{URLLC}{ultra-reliable low-latency communications}
    \acro{V2V}{vehicle-to-vehicle}    
    \acro{wpt}{wireless power transfer}
\end{acronym}

\usepackage{graphicx,wrapfig,lipsum}
\usepackage{rotating}
\usepackage{amsmath, amssymb, amsthm}
\usepackage{stmaryrd}
\usepackage{tikz}
\usepackage{verbatim}
\usepackage{url}
\usepackage[utf8]{inputenc}
\usepackage[english]{babel}
\newtheorem{theorem}{Theorem}[]
\newtheorem{corollary}{Corollary}[]

\newtheorem{example}{Example}[]
\newtheorem{lemma}[]{Lemma}

\usepackage{multicol}

\usepackage{algorithm}
\usepackage{algpseudocode}
\usepackage{scalerel}
\usepackage{multicol}
\usepackage{setspace}

\usepackage[noadjust]{cite}
\usepackage{booktabs}

\usepackage{bbm}
\usepackage[normalem]{ulem}
\usepackage{color}
\usepackage{dsfont}
\usepackage{bm}
\usepackage{setspace}
\usetikzlibrary{automata}
\usetikzlibrary{arrows}
\usetikzlibrary{shapes.geometric, arrows}
\usepackage[]{footmisc}
\usetikzlibrary{positioning}
\usetikzlibrary{calc}
\usetikzlibrary{arrows}
\hypersetup{nolinks=true}
\allowdisplaybreaks

\usepackage{caption}
\usepackage{mathtools}
\begin{document}
	\title{\fontsize{22.8}{27.6}\selectfont Best Arm Identification Based Beam Acquisition in Stationary and Abruptly Changing Environments}
	\author{
 Gourab Ghatak, {\it Member, IEEE}
 \thanks{The author is with the Department of Electrical Engineering at the Indian Institute of Technology (IIT) Delhi, New Delhi, India 110016. Email: gghatak@ee.iitd.ac.in. The author would like to thank Dr. Syamantak Das from IIIT-Delhi for constructive discussions regarding the \texttt{K-SHES} algorithm.}
 \vspace{-1cm}   }
\date{}

\maketitle

\begin{abstract}
    We study the initial beam acquisition problem in millimeter wave (mm-wave) networks from the perspective of best arm identification in multi-armed bandits (MABs). {For the stationary environment, e.g., in mm-wave backhaul/midhaul and fixed wireless access, we propose a novel algorithm called concurrent beam exploration, \texttt{CBE}, in which multiple beams are grouped based on the beam indices and are simultaneously activated to detect the presence of the user.} The best beam is then identified using a Hamming decoding strategy. For the case of orthogonal and highly directional thin beams, we characterize the performance of \texttt{CBE} in terms of the probability of missed detection and false alarm in a beam group (BG). Leveraging this, we derive the probability of beam selection error and prove that \texttt{CBE} outperforms the state-of-the-art strategies in this metric.
    
    Then, for the abruptly changing environments, e.g., in the case of moving blockages, we characterize the performance of the classical sequential halving (\texttt{SH}) algorithm. In particular, we derive the conditions on the distribution of the change for which the beam selection error is exponentially bounded. In case the change is restricted to a subset of the beams, we devise a strategy called $K$-sequential halving and exhaustive search, \texttt{K-SHES}, that leads to an improved bound for the beam selection error as compared to \texttt{SH}. This policy is particularly useful when a near-optimal beam becomes optimal during the beam-selection procedure due to abruptly changing channel conditions. Finally, we demonstrate the efficacy of the proposed scheme by employing it in a tandem beam refinement and data transmission scheme.
\end{abstract}

\section{Introduction}
\subsection{Context and Background}
The \ac{mm-wave} spectrum offers large bandwidths, enabling high data rates for future wireless applications~\cite{niu2015survey}. However, it is highly susceptible to path loss and blockages due to the shorter wavelength~\cite{niu2014blockage}. To overcome such detrimental issues, the \ac{mm-wave} transceivers employ beamforming using large antenna arrays~\cite{marzi2016compressive}. Consequently, in \ac{mm-wave} communication systems, initial beam selection plays a crucial role in establishing a reliable and high-quality link between the \ac{BS} and the \ac{UE}~\cite{giordani2016initial}~\cite{cousik2021fast}. In the case of beam-selection error or beam misalignment, the received signal quality at the \ac{UE} deteriorates significantly, thereby rendering communication infeasible~\cite{ghatak2020beamwidth}.

{The beam selection error naturally increases with the number of available beams or the size of the beam codebook. Similarly, it also increases in dynamic environments. The above two aspects are the focus of this work, i.e., fast initial access in the case of a large number of beams and initial access procedures in a non-stationary environment. In this regard,} note that the number of beams per \ac{SSB} can vary depending on the specific deployment and configuration of the network~\cite{lin2018ss}. An \ac{SSB} burst allows up to a max of 64 \ac{SSB} beams~\cite{ahmadi20195g}. The exact number of beams per \ac{SSB} is determined by the network operator and can be adjusted based on factors such as coverage requirements, network capacity, and radio resource management strategies. 3GPP specifies that multiple beams can be formed and transmitted by the \ac{BS} to cover different areas or sectors~\cite{aykin2019multi}, which we also assume in the first part of our work.

Let us first discuss the related work on different initial access techniques, with a special focus on \ac{MAB} algorithms.

\subsection{Related Work}
{\bf On different initial access techniques:} To perform initial beam acquisition in \ac{mm-wave} systems, researchers have investigated several technologies such as {beam sweeping~\cite{wang2009beam, sweep1, sweep2}, channel estimation~\cite{ghatak2020beamwidth, abu2016random}, compressed sensing~\cite{sung2020compressed, cs2}, hybrid beamforming~\cite{soleimani2019fast, jiang2022initial}, and machine-learning~\cite{cousik2022deep, sohrabi2021deep, zia2019machine, liu2020user}.} Beam sweeping involves transmitting signals using different {beamforming} directions over a predefined set of beams. The receiver then measures the received signal quality for each beam and reports it back to the transmitter. Based on this feedback, the transmitter selects the beam with the highest received signal strength or quality. An exhaustive search of the beam space is associated with high overheads and leads to high initial access delays. To overcome this, researchers have proposed compressed sensing methods, where the \ac{BS} sends a compressed version of the beam codebook to the \ac{UE}~\cite{sung2020compressed}. {For example, the work \cite{sweep1} executes a 3-dimensional peak-finding algorithm to find the best beam in logarithmic time. Then the beamfinding problem is formulated as a sparse problem, and compressive sensing is employed to determine the minimum number of measurements needed for this process. In many recent works, hybrid beamforming~\cite{soleimani2019fast} is exploited for rapid initial access. For example, \cite{jiang2022initial} proposed a novel omnidirectional broadcasting scheme for initial access in \ac{mm-wave} and THz systems with hybrid beamforming. Similarly, \cite{hyb1} proposed fast initial access methods for \ac{mm-wave} systems that exploit hybrid beamforming.}

On the contrary, several approaches for initial access involve the estimation of the channel characteristics between the transmitter and receiver~\cite{jiang2021high, weng2023wideband}. By exploiting the estimated channel information, such as arrival and departure angles, the transmitter can make informed decisions regarding initial beam selection. This falls under the larger umbrella of localization-assisted initial access~\cite{ghatak2020beamwidth}. Recently, machine learning algorithms have been utilized to learn and predict the optimal beam selection based on various channel and environment parameters~\cite{echigo2021deep}. In particular, by training models with large datasets, the transmitter can predict the best beamforming parameters for a given set of conditions, reducing the need for exhaustive beam search procedures~\cite{yang2020machine}.


Alkhateeb {\it el al.}~\cite{Alkhateeb17initial} designed an initial beam association method based on beam sweeping and downlink control pilot reuse. Typically, hierarchical and multi-resolution codebooks result in reduced initial access delays. In this regard, Wang et al.~\cite{wang2009beam} devised an efficient multi-resolution beam search technique that initiates with wide beams and progressively narrows them down until identifying the optimal beam. Nevertheless, the beam resolution requires adjustment at each stage. Wu {\it et al.}~\cite{wu2019fast} presented a technique for rapid and precise beam alignment in multi-path channels within a point-to-point \ac{mm-wave} system. The method capitalizes on the correlation structure among beams, extracting information from neighboring beams to identify the optimal beam efficiently rather than searching through the entire beam space.

The specific method chosen for initial beam selection depends on the system requirements, available resources, and implementation constraints. Beam training and selection are iterative processes, and continuous adaptation may be necessary to maintain an optimal link in dynamic \ac{mm-wave} environments. However, the issue of a changing environment during the beam selection procedure is largely unaddressed. {Since an exhaustive discussion of all initial access procedures for mm-wave systems is out of the scope of this paper, we refer the reader to the work by Giordani {\it et al.}~\cite{giordani2016initial} which provides a comprehensive overview for the same.}

{\bf On \ac{MAB} algorithms for initial access:} Recently, \ac{MAB} frameworks have been employed to study the problem of efficient initial access. It is interesting to note that the two settings of the best arm identification problem in \acp{MAB}, the fixed budget setting and the fixed confidence setting, correspond to two beam acquisition requirements: {a fixed beam selection deadline and a fixed beam selection error. To elaborate further, in a fixed-confidence setting, an algorithm aims to identify the best arm with a pre-specified confidence level. The algorithm keeps sampling arms until it is confident enough in its estimate that one of the arms is the best. Thus, it might take a variable number of samples for each arm and a non-deterministic horizon to decide. On the contrary, in a fixed-budget setting, an algorithm has a predetermined budget for the number of samples that it can use for exploration~\cite{audibert2010best}. The exploration phase stops once the budget is exhausted, and the algorithm selects the arm that appears to be the best based on the collected data. This setting is more concerned with making the best decision within a limited resource constraint (called the fixed budget) than achieving a specific confidence level. From a wireless network perspective, which is characterized by predetermined frame lengths, access slots, and resource constraints, we propose that the fixed-budget setting is more suitable from a protocol perspective.} Hashemi {\it et al.}~\cite{hashemi2018efficient} have studied contextual bandits for beam alignment. Specifically, they address an online stochastic optimization scenario where the objective is to maximize the directivity gain of the beam alignment policy over a specific time frame. By leveraging the inherent correlation and unimodality properties of the model, the authors illustrate that the inclusion of contextual information enhances performance. The work by Va {\it et al.}~\cite{va2019online} utilized a UCB-based framework to create an online learning algorithm for selecting and refining beam pairs. The algorithm initially learns coarse beam directions from a predefined beam codebook and subsequently refines the identified directions to align with the power angular spectrum's peak at that specific position. Hussain {\it et al.}~\cite{hussain2019second} developed an innovative scheme for beam pair alignment utilizing Bayesian \acp{MAB}. The primary objective of this scheme was to maximize both the alignment probability and the throughput of data communication. More recently, Wei~{\it et al.} developed a bandit-based initial beam selection algorithm named two-phase heteroscedastic track-and-stop (\texttt{2PHT\&S})~\cite{wei23pure}. The authors formulated the beam selection as a fixed-confidence pure exploration problem. The authors assumed a correlation structure among beams, considering that the information from nearby beams is similar. Additionally, the algorithm exploits the heteroscedastic property that the variance of the reward of an arm is related to its mean. \texttt{2PHT\&S} groups all beams into several beam sets such that the optimal beam set is first selected and the optimal beam is identified in this set.

\subsection{Motivation}
In almost all of the research above, the authors did not provide any insight into the performance of their algorithms in a non-stationary environment. Our formulation also considers the heteroscedastic Gaussian distribution, and for the stationary environment, we demonstrate that for highly directional thin beams, \texttt{CBE} outperforms \texttt{2PHT\&S}. Additionally, we investigate an algorithm tuned to a changing environment called the \texttt{K-SHES}.

\subsection{Contributions and Organization}
The main contributions in this work are as follows:
\begin{itemize}
    \item {For the stationary environment (e.g., wireless backhaul and midhaul under static blockages (e.g., buildings), wireless HDMI, etc.), we propose and characterize a novel initial beam acquisition algorithm, concurrent beam-exploration (\texttt{CBE}).} The main innovation in \texttt{CBE} is the formation of \acp{BG} based on the beam indices, followed by concurrent multi-beam detection to identify the \acp{BG} in which the \ac{UE} is present. Then, the index of the best beam is decoded for service.
    \item We prove that in the case of highly directional beams that are characterized by negligible sidelobe gains (e.g., see \cite{foo2015orthogonal}), the detection statistic reduces to a generalized Chi-square distributed random variable. We derive the probability of missed detection and the probability of false alarms for the \acp{BG}. Leveraging this, we prove that \texttt{CBE} reduces the probability of beam selection error as compared to the state-of-the-art hierarchical beam selection procedures. {In fact, the derived bound is lower than the lower bound on all algorithms that restrict the framework to single-arm selection per time slot.}
    \item In the case of intermittent blockages, i.e., when a previously suboptimal beam becomes optimal during the beam selection procedure, the performance of \texttt{CBE} deteriorates significantly. For this piece-wise stationary environment, we characterize the performance of the sequential halving (\texttt{SH}) algorithm, which is popular for best arm identification in bandit environments. {To the best of our knowledge, no known algorithms exist that cater to changing environments during an \ac{SSB} burst. Furthermore, from a theoretical standpoint, ours is the first work that rigorously characterizes the performance of \texttt{SH} in an abruptly changing environment.} We show that the upper bound of \texttt{SH} consists of an exponential term and a term dependent on the distribution of the location of the change. Accordingly, we derive conditions on the distribution of the change in order to guarantee an exponential bound for \texttt{SH} in the presence of a single change.
    \item For the case when the change occurs in one of the best $K$ arms, we propose a novel algorithm called $K$-sequential halving followed by exhaustive search (\texttt{K-SHES}) and demonstrate that it outperforms not only \texttt{SH} but also other state-of-the-art algorithms for initial beam acquisition. We also highlight its limitations, specifically for cases of early change.
    \item Finally, as a case study to test the efficacy of \texttt{K-SHES}, we employ it in a tandem beam refinement and data communication system. Based on that, we derive the system design rules for selecting an optimal beam dictionary size and the optimal fractional resources allotted to the beam refinement phase of the system.
\end{itemize}
The rest of the paper is organized as follows: We introduce the system model and define the problem statement for both the stationary and non-stationary cases in Section~\ref{sec:SM_PS}. We focus on the stationary environment in Section~\ref{sec:IESE} and propose and characterize the \texttt{CBE} algorithm. The abruptly changing environment is considered in Section~\ref{sec:BBSACE}. The heuristic hybrid policy \texttt{K-SHES} is proposed in Section~\ref{sec:HPKK}. Some numerical results and the case study are discussed in Section~\ref{sec:NRD}. Finally, the paper concludes in Section~\ref{sec:C}.

\section{System Model and Problem Statement}
\label{sec:SM_PS}
We consider the propagation environment with limited scattering (typical for \ac{mm-wave} channels) and adopt the commonly-used geometric channel model~\cite{wei23pure}. Let us consider a \ac{ULA}, however, it will shortly be apparent that the framework can be applied to a \ac{UPA} since the analysis follows only from the beam directions. The beamforming codebook $\mathcal{N}$ of size $N$ is
    {
\begin{align}
    C \triangleq \{{\bf f}_i = {\bf a}\left(-1 + 2i/N\right) \mid i = 0, 1, \ldots, N-1\},
\end{align}
where ${\bf a}(\cdot)$ denotes the array response vector. The structure of ${\bf a}(\cdot)$ for \ac{ULA} can be found in~\cite{wei23pure} and is skipped here for brevity.}
The received signal in case only ${\bf f}_i$ is activated is
\begin{align}
    y_i = \sqrt{P}{\bf h}_i^H{\bf f}_i + n, \nonumber 
\end{align}
where ${\rm h}$ is the channel vector.
Thus, the received power is
\begin{align}
    R_i = P|{\bf h}_i^H{\bf f}_i|^2 + |n|^2 + 2\sqrt{P}\mathbb{R}\left({\bf h}_i^H{\bf f}_i n^1\right), \nonumber 
\end{align}
where $\mathbb{R}(\cdot)$ denotes the real part of the argument.
Since the noise power is negligible, the received power is Gaussian distributed {with mean downlink power} $\mu_i = P|{\bf h}_i^H{\bf f}_i|^2$ and variance $\sigma_i^2 = 2P|{\bf h}_i^H{\bf f}_i|^2\sigma^2 = 2\sigma^2 \mu_i$~\cite{wei23pure}. In what follows, we formulate two problem statements, ${\rm P}_1$ and ${\rm P}_2$, for the stationary and the abruptly changing cases, respectively. In addition, for the stationary environment, we make the additional assumption that the best beam, i.e., the beam in which the user is aligned, has a gain of $G$, while all the other beams, i.e., the ones not aligned towards the \ac{UE} have a gain of $g$. Note that this assumption is only for the stationary environment, while for the abruptly changing environment, the model is more general, as described later.

\subsection{Stationary Environment Problem}
{First, we focus on use cases such as fixed wireless access, including wireless backhaul and midhaul under static blockages (e.g., buildings), wireless hi-definition multimedia connectivity (e.g., with wireless HDMI) in indoor environments, etc. Scenarios with low-mobility users that do not contain dynamic blockages are characterized by such a stationary environment. The problem of the best beam identification in such cases} is the same as selecting the beam with the highest $\mu_i$ within a beam selection deadline $T$.
\begin{align}
    {\rm P}_1:\qquad& \text{Find } \argmax_i \mu_i, \nonumber \\
    {\rm st} \qquad&  R(t) \sim \mathcal{N}\left(\mu_i, \sigma_i\right), \quad \forall i, \forall t \in [T] \nonumber \\
   {\rm within} \qquad& T. \nonumber 
\end{align}
{Here $R(t)$ is the downlink power at the user at time slot $t$ given that the beam-search policy activates the beam ${\bf f}_i$ at $t$.}
The critical challenge is the fact that the beam with the highest $\mu_i$ is also the same as the beam with the highest $\sigma_i$, and accordingly, a higher number of samples is needed to estimate $\mu_i$. At the end of $T$, let the selected beam by an algorithm $\mathcal{Z}$ be $f_{\mathcal{Z}}$. Then the probability of beam selection error is given as
\begin{align}
    \mathcal{P}^{\mathcal{Z}}_{\rm e} = \mathbb{P}\left(f_{\mathcal{Z}} \neq \max \{{\bf f}_i\}\right), \nonumber 
\end{align}
where the subscript ${\rm e} \in \{{\rm NC}, {\rm C}\}$ stands for either a stationary (${\rm NC}$: no-change) or non-stationary (${\rm C}$: changing) environment. The typical benchmark used for comparing proposed beam selection algorithms is the exhaustive search~\cite{li2017design}, where each beam is activated sequentially and, based on multiple measurements for each beam, the best beam is selected. Other popular algorithms with which we compare our results are hierarchical search~\cite{8962355} for the stationary and non-stationary cases and a variant of UCB-E and successive rejects~\cite{audibert2010best} for the non-stationary case. First, let us note that the beam selection error for the exhaustive search approach is obvious from the Chernoff bound~\cite{hellman1970probability} as follows:
\begin{lemma}
    For single beam activation, with $T_{i} \in \mathbb{N}$ observations, the estimate $\hat{\mu}_i$ of $\mu_i$ is $\epsilon$ close to $\mu_i$ is given by
    \begin{align}
        \mathbb{P}\left(|\hat{\mu}_i - \mu_i| \geq \epsilon\right) \leq \exp\left(- \frac{T_i \epsilon^2}{4\sigma_i^2 \mu_i}\right). \nonumber 
    \end{align}
\end{lemma}
In the case of fixed-budget exploration, equal temporal resources are allotted to each beam, i.e., $T_i = \frac{T}{N}$.  Thus, in case of an exhaustive search in a stationary environment, the probability of beam-selection error is upper bound as
\begin{align}
    \mathcal{P}^{\rm ES}_{\rm NC}(T) &\leq \sum\limits_{i \neq 1} \mathbb{P}\left(|\hat{\mu}_i - \hat{\mu}_{1}| \geq \epsilon\right)   \nonumber \\
    &\leq N\exp\left(- \frac{T \Delta^2_{\min}}{8N\sigma^2 \mu_{\max}}\right),
\end{align}
{where $\Delta_{\min}$ and $\mu_{\max}$ are the minimum sub-optimality gap and the mean downlink power of the best beam, respectively.}
In the first part of the paper, i.e., in Section~\ref{sec:IESE}, we propose a grouped exploration strategy that outperforms not only this benchmark but also the popular hierarchical search algorithm for highly directional beams.

\subsection{Abruptly Changing Environment Problem}
In the second part of this paper, i.e., in Section IV, {we consider mobile user scenarios where the blockage conditions are highly dynamic, e.g., aerial and terrestrial vehicular communications, AR/VR-based gaming, etc. This requires the development of beam-acquisition protocols with possible environment changes during an \ac{SSB} burst.} At a time slot $t$ within the beam-selection deadline $T$, the mean power of a suboptimal beam ${\bf f}_j$ changes from $\mu_j^-$ to $\mu_j^+$ so that it becomes optimal for all time slots beyond $t$. This is typical in cases where the optimal beam is blocked during the initial parts of the beam selection process and the blockage shifts during the beam selection process. {This may also be a consequence of changing user orientation with respect to the temporally best beam. It must be noted that other changing environments are feasible, e.g., a gradual drift of means due to user mobility. However, we restrict our discussion to only abrupt changes.} The challenge here is to still identify the best beam at the beam selection deadline $T$, as described below.
\begin{align}
    {\rm P}_2:\qquad& \argmax_i \mu_i(T), \nonumber \\
    {\rm st} \qquad&  R(t) \sim \mathcal{N}\left(\mu_i(t), \sigma_i(t)\right), \quad \forall i, \forall t \in [T] \nonumber \\
    &\mu_i(t) = \mu_i, \quad \forall i\neq j, \forall t, \nonumber \\
    &\mu_j(t) = \mu_j^-, \quad  0 \leq t \leq t_{\rm c}, \nonumber \\
    &\mu_j(t) = \mu_j^+, \quad  t_{\rm c} < t \leq T, \nonumber \\
   {\rm within} \qquad& T. \nonumber 
\end{align}
Note that the parameters for all other beams except ${\bf f}_j$ remain constant.

\begin{table*}[t]
\centering
\begin{tabular}{ |c |c | }
\hline
\textbf{Symbol} &\textbf{Definition} \\
\hline 
$\mathcal{N}$, $N$ &Set of beams, number of beams. \\
${\bf f}_i$ &$i-$th beam. \\
$T$ & Beam selection deadline. \\
$\mu_i(t)$ &Mean power of the $i-$th beam at time $t$. \\
$\mu_j^-$ &Mean power of the $j-$th beam before change. \\
$\mu_j^+$ &Mean power of the $j-$th beam after change. \\
$\Delta_{\rm c}$ &Change in the means, $\Delta_{\rm c} = \mu_j^+ - \mu_j^-$.\\
$\Delta^+_{i,j}$ &$\mu_i - \mu_j^+$.\\
$\Delta^-_{i,j}$ &$\mu_i - \mu_j^-$.\\
$g, G$ &Gain of the incorrect and best beams, respectively, in the stationary case. \\
$B_k$ & $k-$th \ac{BG}. \\
${\bf R}_{B_k}$ & Observations/Rewards from the $k-$th \ac{BG} in \texttt{CBE}. \\
$\mu_1, \sigma^2_1$ & The mean and the variance of ${\bf R}_{B_k}$ given the user is present in the $k-$th \ac{BG} for \texttt{CBE}. \\
$\mu_0, \sigma^2_0$ & The mean and the variance of ${\bf R}_{B_k}$ given the user is not present in the $k-$th \ac{BG} for \texttt{CBE}. \\
$K$ & Maximum index of the beam which undergoes a change. \\
$t_{\rm c}$ &Time slot of change. \\
$n_{r_{\rm c}}$ &Change round in the \texttt{SH} algorithm. \\
$n'$ &Change slot conditioned on the change round. \\
$F_{n'}(n | n_{r_{\rm c}})$ &conditional {CDF} of $n'$ given $r_{\rm c}$. \\
$\Delta_{\rm min}$ & Minimum difference between the means of beams. $\Delta_{\rm min} = 0$ for \texttt{CBE}.\\
$\mu_{\rm max}$ & Maximum mean power of a beam. \\
$\sigma^2_{\rm max}$ & $2 \sigma^2 \mu_{\rm max}$.\\
\hline
\end{tabular}
\caption{List of notations}
\label{tab:notations}
\end{table*}

\section{Indexed Exploration for Stationary Environment}
Let us first analyze the stationary environment. {We propose an algorithm based on concurrent multi-beam transmission~\cite{liu2022concurrent}, where, in each \ac{SSB} beam slot, multiple beam directions are simultaneously activated to construct a single beam pattern. An \ac{SSB} burst allows up to a max of 64 \ac{SSB} beams~\cite{ahmadi20195g}. Thus, with a fixed deadline, a particular beam direction is probed a higher number of times as a part of different \acp{BG} as explained subsequently.}
The first step for \texttt{CBE} is to form \acp{BG}, as discussed below.
\label{sec:IESE}
\subsection{Beam Grouping}
The $i$-th beam ${\bf f}_i$ is added to the \ac{BG} $B_k$, $k = 1, 2, \ldots, d$, if and only if the binary representation of $i$ has a "1" in the $k$-th binary place. In other words, ${\bf f}_i$ is added to $B_k$ if
 $   \texttt{bin2dec}\left(\texttt{dec2bin}(i) \; \texttt{ AND } \; \texttt{onehot}(k) \right) \neq 0$,
where $\texttt{bin2dec}()$ and $\texttt{dec2bin}()$ are respectively operators that convert binary numbers to decimals and decimal numbers to binary. Additionally, $\texttt{onehot}(k)$ is a binary number with all zeros except 1 at the $k$-th binary position. $\texttt{AND}$ is the bit-wise AND operator.
{The beam grouping strategy imitates the parity bit generation strategy of Hamming codes, which belong to the family of linear error-correcting codes~\cite{hamming1950error}. The key observation is that Hamming codes are perfect codes; that is, they achieve the highest possible rate for codes with a given block length and minimum distance of three. The minimum distance criterion allows not only the detection but also the identification and correction of single-bit errors. The corollary is that the minimum number of parity bits needed to detect and correct one-bit errors is $\log_2 (K)$ for a $K$ bit sequence. This directly implies that the minimum number of groups that are needed to detect the presence of the user in a single beam as well as identify the corresponding beam is $\log_2 (K)$. In other words, here the presence of the user is analogous to a bit-error and an error detection by a parity group is equivalent to the UE detection by a \ac{BG}.} Such a grouping strategy was explored in \cite{ghataktsge} for fast detection of changes in a classical bandit environment. However, here we leverage the same for quick identification of the best beam.

 {\bf Example:} Let us elaborate this further with an illustrative example by considering $N = 16$. Fig.~\ref{fig:beamgroups} shows the following grouping for the beams: i) $B_1$: beams with '1' in the first binary place: {${\bf f}_1, {\bf f}_3, {\bf f}_5, {\bf f}_7, {\bf f}_9, {\bf f}_{11}, {\bf f}_{13},$ and ${\bf f}_{15}$}, ii) $B_2$: beams with '1' in the second binary place: ${\bf f}_2, {\bf f}_3, {\bf f}_6, {\bf f}_7, {\bf f}_{10}, {\bf f}_{11}, {\bf f}_{14},$ and ${\bf f}_{15}$, iii) $B_3$: beams with '1' in the third binary place: ${\bf f}_4, {\bf f}_5, {\bf f}_6, {\bf f}_7, {\bf f}_{12}, {\bf f}_{13}, {\bf f}_{14},$ and ${\bf f}_{15}$, and iv) $B_4$: beams with '1' in the fourth binary place: ${\bf f}_8, {\bf f}_9, {\bf f}_{10}, {\bf f}_{11}, {\bf f}_{12}, {\bf f}_{13}, {\bf f}_{14},$ and ${\bf f}_{15}$.
\begin{figure}
    \centering
    \includegraphics[width = 0.5\textwidth]{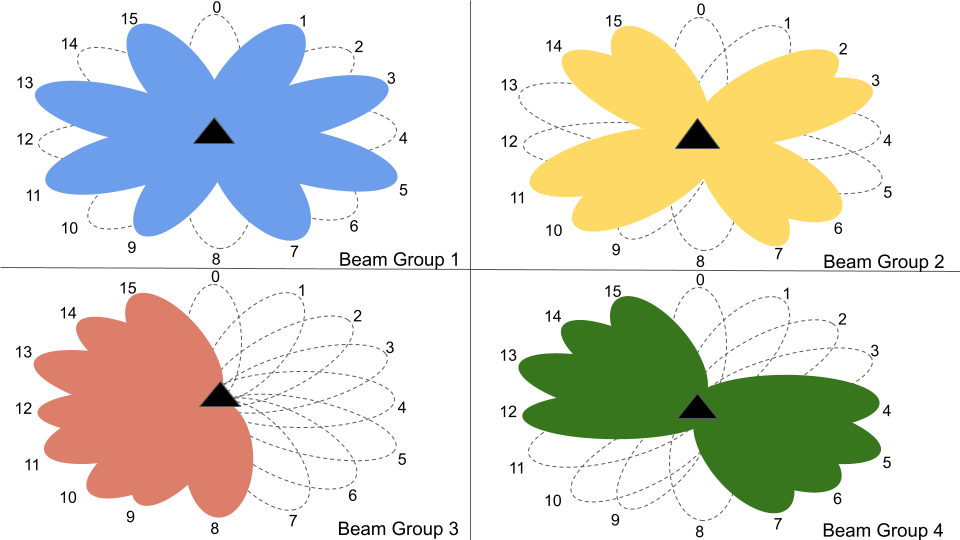}
    \caption{Illustration of \ac{BG}}
    \label{fig:beamgroups}
\end{figure}

\subsection{Rewards}
In case the \ac{BG} $B_k$ is employed to measure the downlink power, the received power is given as
\begin{align}
 &P_{B_k} =\sum_{{\bf f}_i \in B_k} \frac{2P_{\rm t}}{N} |{\bf h}^H {\bf f}_i|^2 + |n|^2 + \nonumber \\
&\mathbb{R}\left(2\sum_{{\bf f}_i \in B_k} \sqrt{\frac{P_{\rm t}}{\frac{N}{2}}}{\bf h}^H {\bf f}_i n^1 + \nonumber \right.\left.2\sum_{{\bf f}_i \in B_k}\sum_{{\bf f}_j \in B_k, {\bf f}_j \neq {\bf f}_i} \frac{P_{\rm t}}{\frac{N}{2}} {\bf f}_i^H {\bf h} {\bf h}^H{\bf f}_j\right),
\end{align}
{Note that with multi-beam transmission, the radiated power per beam is reduced since the power is distributed among the selected beam directions. Specifically, since we activate $\frac{N}{2}$ beams in each \ac{BG}, we assume that the power reduces to $\frac{P_{\rm t}}{N/2}$.}
Similar to~\cite{wei23pure}, we assume that the noise variance is much smaller than the transmit power. {Additionally, since we assume highly directional beams~\cite{foo2015orthogonal}, we neglect the contribution of ${\bf f}_i^H {\bf h} {\bf h}^H{\bf f}_j$ as the magnitudes of either $|{\bf f}_i^H {\bf h}|$ or $|{\bf f}_j^H {\bf h}|$ are low for $i \neq j$.} Accordingly, the variable $P_{B_k}$ is approximately a Gaussian random variable with mean $\mu_{B_k} =  \sum_{{\bf f}_i \in B_k}\frac{P_{\rm t}}{\frac{N}{2}} |{\bf h}^H {\bf f}_i|^2 = \sum_{{\bf f}_i \in B_k} \frac{2\mu_i P_t}{N}$ and variance $\sigma^2_{B_k} = 2\sum_{{\bf f}_i \in B_k}\frac{P_{\rm t}}{\frac{N}{2}} |{\bf h}^H {\bf f}_i|^2 \sigma^2 = \sum_{{\bf f}_i \in B_k}\frac{2\mu_i P_t \sigma^2}{ N}$.

\subsection{Beam Selection Strategy}
The beam selection strategy for \texttt{CBE} is summarized in Algorithm~\ref{alg:beams}. We divide the total initial access time into $\log N$ rounds and allot each round to one \ac{BG}\footnote{All logarithms in this paper unless otherwise stated have a base 2.}. Then, all the beams of a \ac{BG} are activated to detect the presence of the user in that particular \ac{BG}. Assume that $\mathds{1}(B_k)$ indicates the presence of the user in \ac{BG} $B_k$. Then, \ac{UE} detection is based on the classical likelihood ratio test. Note that the conditional \acp{PDF} of ${\bf R}_{B_k}$ given that the user is respectively present and absent in the \ac{BG} $B_k$, are
\begin{align}
    f_{{\bf R}_{B_k}\mid \mathds{1}\left(B_k\right)}\left( {\bf y} \mid 1 \right) &= \prod\limits_{j = 1}^{T_{B_k}} \frac{1}{\sigma_1^2\sqrt{2\pi}} \exp\left[ -\frac{\left(y_j - \mu_1\right)^2}{2\sigma_1^2} \right], \nonumber \\
    f_{{\bf R}_{B_k}\mid \mathds{1}\left(B_k\right)}\left( {\bf y} \mid 0 \right) &= \prod\limits_{j = 1}^{T_{B_k}} \frac{1}{\sigma_0^2\sqrt{2\pi}} \exp\left[ -\frac{\left(y_j - \mu_0\right)^2}{2\sigma_0^2} \right]. \nonumber 
\end{align}
Here,
\begin{subequations}
\begin{equation}
    \mu_0 = g,\\
    \label{eq:mu0}
\end{equation}
\begin{equation}
    \mu_1 = \frac{2}{N}\left(\left(\frac{N}{2} - 1\right)g + G\right),\\
        \label{eq:mu1}
\end{equation}
\begin{equation}
    \sigma^2_0 = 2g\sigma^2,\\
        \label{eq:sigma0}
\end{equation}
\begin{equation}
    \sigma^2_1 = \frac{4\sigma^2}{N}\left(\left(\frac{N}{2} - 1\right)g + G\right).\\
    \label{eq:sigma1}
\end{equation}
\label{eq:values}
\end{subequations}
Accordingly, the \ac{LLR} is evaluated as
\begin{align}
    {\rm LLR}({\bf R}_{B_k}) &= \frac{\sigma_0}{\sigma_1} + \left[\sum\limits_{j = 1}^{T_{B_k}} -\frac{\left(y_j - \mu_1\right)^2}{2\sigma_1^2} + \frac{\left(y_j - \mu_0\right)^2}{2\sigma_0^2} \right] \nonumber \\
    & = \frac{\sigma_0}{\sigma_1} + \left[\sum\limits_{j = 1}^{T_{B_k}} \frac{1}{\sigma_0^2 \sigma_1^2} \left( y_j^2\left(\sigma_1^2 - \sigma_0^2\right) + 
\right.\right. \nonumber \\
& \left.\left.2y_j\left(\mu_1\sigma_0^2 - \mu_0 \sigma_1^2\right) + \left(\sigma_1^2\mu_0^2 - \mu_1^2\sigma_0^2\right) \right)\right], \nonumber
\end{align}
where, interestingly, substituting \eqref{eq:values}, we get $\mu_1\sigma_0^2 - \mu_0 \sigma_1^2 = 0$. Thus,
\begin{align}
   {\rm LLR}({\bf R}_{B_k}) & =  \frac{\sigma_0}{\sigma_1} + T_{B_k}\left(\frac{\mu_0^2}{\sigma_0^2} - \frac{\mu_1^2}{\sigma_1^2}\right) +\sum\limits_{j = 1}^{T_{B_k}} y_j^2 \frac{\sigma_1^2 - \sigma_0^2}{\sigma_1^2 \sigma_0^2} \nonumber \\
   & = \sqrt{\frac{{Ng}}{2g'}} - T_{B_k} \frac{G - g}{N\sigma^2} +   \frac{\frac{2}{N}g' - g}{\frac{4g}{N}\sigma^2g'}\sum\limits_{j = 1}^{T_{B_k}} y_j^2, \nonumber 
\end{align}
where
\begin{align}
    g' = \left(\frac{N}{2} - 1\right)g + G \nonumber.
\end{align}
Conveniently, our test statistic and the decision rule for \ac{BG} $B_k$ is
\begin{align}
    \norm{{\bf R}_{B_k}}^2 = \sum\limits_{j = 1}^{T_{B_k}} y_j^2 \underset{\text{user not detected}}{\overset{\text{user detected}}{\gtreqless}} \gamma,
    \label{eq:teststat}
\end{align}
where
\begin{align}
    \gamma = &\frac{4gg'\sigma^2}{\frac{2}{N}g' - g}\left[1 - \sqrt{\frac{Ng}{2g'}} - T_{B_k}\left(\frac{g - G}{N\sigma^2}\right)\right]. \nonumber
\end{align}
Finally, based on the detection of the user in different \acp{BG}, the best beam is identified as the one that belongs to all the \acp{BG} in which the user is detected. For this, let us define a new sequence of sets as
\begin{align}
    C_k = \begin{cases}
    B_k; \quad \text{If the user is detected in }B_k, \\
    {B}_k^{\rm C}; \quad \text{If the user is not detected in }B_k.
    \end{cases}
    \label{eq:Csets}
\end{align}
Then, the optimal beam is identified as ${\bf f}_j$, where ${\bf f}_j = \bigcap\limits_{k = 1}^{\log  N} C_k$. In case the user is not detected in any of the \acp{BG}, the optimal beam is identified as ${\bf f}_0$.

\begin{table}[ht]
    \centering
    \begin{tabular}{|c|c|c|}
    \hline
         User in & Detected &Not detected \\
         \hline
         ${\bf f}_0$ & - &$B_1$, $B_2$, $B_3$ and $B_4$\\
         ${\bf f}_1$ &  $B_1$ &$B_2, B_3,$ and $B_4$\\
         ${\bf f}_2$ &  $B_2$ &$B_1, B_3,$ and $B_4$ \\ 
         ${\bf f}_3$ &  $B_1$ and $B_2$ &$B_3$ and $B_4$\\
         ${\bf f}_4$ &  $B_3$ &$B_1, B_2,$ and $B_4$\\
         ${\bf f}_5$ &  $B_1$ and $B_3$ &$B_2$ and $B_4$\\
         ${\bf f}_6$ &  $B_2$ and $B_3$ &$B_1$ and $B_4$\\
         ${\bf f}_7$ &  $B_1$, $B_2$, and $B_3$ &$B_4$\\
         ${\bf f}_8$ &  $B_4$ &$B_1, B_2,$ and $B_3$\\
         ${\bf f}_9$ &  $B_1$ and $B_4$ &$B_2$ and $B_3$\\
         ${\bf f}_{10}$ &  $B_2$ and $B_4$ &$B_1$ and $B_2$\\
         ${\bf f}_{11}$ &  $B_1$, $B_2$, and $B_4$ &$B_3$\\
         ${\bf f}_{12}$ &  $B_3$ and $B_4$ &$B_1$ and $B_2$\\
         ${\bf f}_{13}$ &  $B_1$, $B_3$, and $B_4$ &$B_2$\\
         ${\bf f}_{14}$ &  $B_2$, $B_3$, and $B_4$ &$B_1$\\
         ${\bf f}_{15}$ &  $B_1$, $B_2$, $B_3$ and $B_4$ &-\\
         \hline
    \end{tabular}
    \caption{Exhaustive list of the combination of the \acp{BG} in which the user is detected.}
    \label{tab:armchanges}
\end{table}
Let us recall the illustration in Fig.~\ref{fig:beamgroups}. Corresponding to this case of 16 beams and 4 \acp{BG}, Table~\ref{tab:armchanges} exhaustively enlists the cases of beam identification. As an example, if the user is detected in $B_1$ but not in any other \ac{BG}, then the beam ${\bf f}_{0}$ is selected for it. Similarly, if the user is detected in all the \acp{BG}, then the beam ${\bf f}_{15}$ is identified as the best beam.

\begin{algorithm}
\caption{\texttt{CBE}}\label{alg:beams}
\begin{algorithmic}[1]
\State {\bf Input:  $[{\bf F}]$}
\State Create \acp{BG} $\left([{\bf F}]\right)$
    \For{$k = 1$ to ${\log  N}$}
        \State - Transmit using $B_k$ for $T_{B_k}$ slots and observe ${\bf R}_{B_{k}}$.
        \State - Detect whether the user is present in \ac{BG} $B_k$ using \eqref{eq:teststat}.
        \State - Create sequence of sets $C_k$ using \eqref{eq:Csets}.
    \EndFor
    \State {\bf Return:} Optimal beam ${\bf f}_j$, where ${\bf f}_j = \bigcap\limits_{k = 1}^{\log  N} C_k$.
\end{algorithmic}
\end{algorithm}

\subsection{Characterization of the Test Statistic}
We note that for the \ac{BG} $B_k$, $\frac{\norm{{\bf R}_{B_k}}^2}{\sigma_l^2}$ has a non-central Chi-squared distribution with ${T_{B_k}}$ degrees of freedom and a non-centrality parameter $T_{B_k} \frac{\mu_l}{\sigma_l^2}$, where $l \in \{0, 1\}$, respectively denoting the presence and the absence of the user in the \ac{BG} $B_k$. Mathematically, if $y_i \sim \mathcal{N}\left(\mu_l, \sigma_l^2\right)$, we have
\begin{align}
    \frac{\norm{{\bf R}_{B_k}}^2}{\sigma_l^2} \sim \chi_{\rm NC}^2\left(T_{B_k}, T_{B_k}\frac{\mu_l^2}{\sigma_l^2}\right),
\end{align}
where $\chi_{NC}^2(a,b)$ is the non-central Chi-squared distribution with $a$ degrees of freedom and non-centrality parameter $b$. Accordingly, the conditional \ac{CDF} of $\frac{\norm{{\bf R}_{B_k}}^2}{\sigma_l^2}$ is
\begin{align}
    F_{\frac{\norm{{\bf R}_{B_k}}^2}{\sigma_l^2} \mid \mathds{1}(B_k) = l}(x) &= \mathbb{P}\left(\frac{\norm{{\bf R}_{B_k}}^2}{\sigma_l^2} \leq x \mid \mathds{1}(B_k) = l\right) \nonumber \\
    &= 1 - \mathcal{Q}_{\frac{T_{B_k}}{2}}\left(\sqrt{T_{B_k}} \frac{\mu_l}{\sigma_l}, \sqrt{x}\right), 
\label{eq:MarcumQCDF}
\end{align}
where 
\begin{align}
&\mathcal{Q}_{\frac{T_{B_k}}{2}}\left(\sqrt{T_{B_k}} \frac{\mu_l}{\sigma_l}, \sqrt{x}\right) =    \frac{1}{\left(\sqrt{T_{B_k}} \frac{\mu_l}{\sigma_l}\right)^{\frac{T_{B_k}}{2} - 1}} \int\limits_{\sqrt{x}}^{\infty} x^{\frac{T_{B_k}}{2}} \cdot \nonumber \\
&\exp\left(- \frac{x^2 + \left(\sqrt{T_{B_k}} \frac{\mu_l}{\sigma_l}\right)^2}{2}\right) \mathcal{I}_{\frac{T_{B_k}}{2} - 1}\left(\sqrt{T_{B_k}} \frac{\mu_l}{\sigma_l} x\right) dx,
\end{align}
is the Marcum Q-function~\cite{shnidman1989calculation} and $\mathcal{I}_\nu(\cdot)$ is the modified Bessel function of first kind of order $\nu$~\cite{bowman2012introduction}.

\subsection{Probability of Missed Detection}
Missed detection occurs when $ \norm{{\bf R}_{B_k}}^2 = \sum\limits_{j = 1}^{T_{B_k}} y_j^2 \leq \gamma,$ given that the user is present in the \ac{BG} $B_k$. The probability of missed detection is evaluated as
\begin{align}
    p_{\rm m} &= \mathbb{P}\left(\norm{{\bf R}_{B_k}}^2 \leq \gamma \mid \mathds{1}\left(B_k\right) = 1 \right) \nonumber \\
    & {=} \mathbb{P}\left(\frac{\norm{{\bf R}_{B_k}}^2}{\sigma_1^2} \leq \frac{\gamma}{\sigma_1^2} \mid \mathds{1}\left(B_k\right) = 1 \right) \nonumber \\
    & = 1 - \mathcal{Q}_{\frac{T_{B_k}}{2}}\left(\sqrt{T_{B_k}} \frac{\mu_1}{\sigma_1}, \frac{\sqrt{\gamma}}{\sigma_1}\right).
\end{align}
Next, consider the arguments of the Marcum Q-function above as $a_1 = \sqrt{T_{B_k}} \frac{\mu_1}{\sigma_1}$ and $b_1 = \frac{\sqrt{\gamma}}{\sigma_1^2}$, respectively. Thus, we have
\begin{align}
    a_1^2 &= \frac{g'T_{B_k}}{N\sigma^2}, \nonumber \\
    b_1^2 &= \frac{\frac{4gg'\sigma^2}{\frac{2}{N}g' - g}\left[1 - \sqrt{\frac{Ng}{2g'}} - T_{B_k}\left(\frac{g - G}{N\sigma^2}\right)\right]}{\frac{4\sigma^2g'}{N}} \nonumber \\
    & = \frac{Ng}{\frac{2}{N}g' -g}\left[1 - \frac{Ng}{2g'} + T_{B_k}\frac{G -g}{\sigma^2 N}\right].\nonumber 
\end{align}

\begin{corollary} [\textbf{Chernoff-type bound}]
    Clearly, $b_1^2 < \frac{T_{B_k}}{2}\left(a_1^2 + 2\right)$ and hence, for $0 < \lambda < \frac{1}{2}$, we can derive a Chernoff-type bound for the probability of missed detection as
\begin{align}
    p_{\rm m} \leq \left(1 - 2\lambda\right)^{-\frac{T_{B_k}}{2}} \exp\left(-\lambda b_1^2 + \frac{\lambda T_{B_k} a_1^2}{2 (1 - 2\lambda)}\right).  \nonumber  
\end{align}
The detailed steps to derive the above bound can be found in \cite{marcumbound} and is being skipped here for brevity. The optimal value for the Chernoff parameter $\lambda$ is found by differentiation as
\begin{align}
    \lambda_1 = \frac{1}{2} \left(1 - \frac{T_{B_k}}{2b_1^2} - \frac{T_{B_k}}{2b_1^2}\sqrt{1 - \frac{2a_1^2b_1^2}{T_{B_k}}}\right). 
\end{align}
\end{corollary}
However, in such cases, $p_{\rm m}$ is trivially upper-bounded by 1. In order to derive a more meaningful bound, we note that 
\begin{align}
    \lim_{g \to 0} a_1^2 = \frac{GT_{B_k}}{2\sigma^2 N}, \qquad \lim_{g \to 0} b_1^2 = 0.
    \label{eq:lim_a_b_1}
\end{align}
Following this observation, we derive the following bound on the probability of missed detection.
\begin{lemma}
\label{cor:pm_bound}
For some $C_{1} \geq 0$,
\begin{align}
    p_{\rm m} \leq C_{1}\exp\left(-\frac{GT}{2N\sigma^2\log  N}\right),\label{eq:CS_md}
\end{align}
where $\lim\limits_{g \to 0} C_1 = 0$ $\forall T$.
\end{lemma}
\begin{IEEEproof}
Please see Appendix~\ref{app:CBE_miss}.
\end{IEEEproof}

\subsection{Probability of False Alarm}
False alarm is raised when $ \norm{{\bf R}_{B_k}}^2 = \sum\limits_{j = 1}^{T_{B_k}} y_j^2 > \gamma,$ given that the user is not present in the \ac{BG} $B_k$. The probability of false alarm is evaluated as
\begin{align}
    p_{\rm f} &= \mathbb{P}\left(\norm{{\bf R}_{B_k}}^2 > \gamma \mid \mathds{1}\left(B_k\right) = 0 \right) \nonumber \\
    & {=} \mathbb{P}\left(\frac{\norm{{\bf R}_{B_k}}^2}{\sigma_0^2} > \frac{\gamma}{\sigma_0^2} \mid \mathds{1}\left(B_k\right) = 0 \right) \nonumber \\
    & = \mathcal{Q}_{\frac{T_{B_k}}{2}}\left(\sqrt{T_{B_k}} \frac{\mu_0}{\sigma_0}, \frac{\sqrt{\gamma}}{\sigma_0}\right).
\end{align}
Letting $a_0 = \sqrt{T_{B_k}} \frac{\mu_0}{\sigma_0}$ and $b_0 = \frac{\sqrt{\gamma}}{\sigma_0^2}$, respectively, we have
\begin{align}
    a_0^2 &= \frac{T_{B_k}g}{2\sigma^2}, \nonumber \\
    b_0^2 &= \frac{2g'}{\frac{2}{N}g' - g}\left[1 - \sqrt{\frac{Ng}{2g'}} + T_{B_k} \frac{(G- g)}{\sigma^2 N}\right]. \nonumber 
\end{align}
Thus, we have
\begin{align}
\lim\limits_{g \to 0} b_0^2 = N + \frac{GT}{\sigma^2 N \log  N}, \quad \lim\limits_{g \to 0} a_0^2 = 0. 
\label{eq:lim_a_b_0}
\end{align}
Following this observation, we derive the following bound on the probability of missed detection.
\begin{lemma}
\label{cor:pf_bound}
For some $C_{2} \geq 0$
\begin{align}
    p_{\rm f} \leq C_{2}\exp\left(-\frac{GT}{\sigma^2 N\log N}\right),  \label{eq:CS_fa}
\end{align}
where $\lim\limits_{g \to 0}C_0 = 0$, $\forall T$.
\end{lemma}
\begin{IEEEproof}
Please see Appendix~\ref{app:pf_bound}.
\end{IEEEproof}
Interestingly, for $p_{\rm f}$, the upper bound is tighter than the one for $p_{\rm m}$ due to the extra $\exp\left(-N\right)$ term in the former. Now we are in a position to state the main result for \texttt{CBE}.
\begin{theorem}
With \texttt{CBE}, the probability of beam-selection error is upper bound as
    \begin{align}
        \mathcal{P}^{\rm CBE}_{\rm NC}(T) \leq L_1  \log  N \exp\left(-\frac{GT}{2N\sigma^2\log  N}\right), \nonumber 
    \end{align}
where $L_1 = \max\{C_1, C_2\}$.
\end{theorem}
\begin{IEEEproof}
The proof follows from Lemma~\ref{cor:pm_bound} and Lemma~\ref{cor:pf_bound}. Note that each beam ${\bf f}_i$ belongs to $v_i$ \acp{BG}, where $v_i = \{1, 2, , \ldots, \log  N\}$. As an example, for $N = 16$, ${\bf f}_1$ belongs to only $\{B_1\}$ and hence, $v_1 = `1$ while ${\bf f}_{15}$ belongs to $\{B_1, B_2, B_3, B_4\}$ and hence $v_{15} = 4$. The number of \acp{BG} a beam ${\bf f}_i$ belongs to depends on the number of '1's in the binary representation of $i$. Accordingly, for the event that the user is present in beam ${\bf f}_i$ for each $i$, the probability of beam selection error is upper bound as
    \begin{align}
        \mathcal{P}^{\rm CBE}_{\rm NC}(T) &\leq \nu_i p_{\rm m} + (N - \nu_i) p_{\rm f} \nonumber \\
        &\leq \log N \max \{p_{\rm m}, p_{\rm f}\} \nonumber \\
        & \leq \log  N \max\{C_1, C_2\} \exp\left(-\frac{GT}{2N\sigma^2\log  N}\right). \nonumber 
    \end{align}  
\end{IEEEproof}

\subsection{Comparison with Other Bounds}
The bound derived in the work by Karnin {\it et al.}~\cite{karnin2013almost} for the probability of best arm selection error as $3 \log  N \exp\left(- \frac{T}{8H_2 \log  N}\right)$, where $H_2 := \max_{i \neq 1} \frac{i}{\Delta_i^2}$, which in our case is $\frac{N}{G - g}$ and hence the bound is $3 \log  N \exp\left(- \frac{T(G-g)}{8N \log  N}\right)$. We note that, similar to hierarchical search, \texttt{SH} also discards half the possible beams for the codebook at each stage. Due to the result that $\lim_{g \to 0} L_1 = 0$, the bound derived by us for this case of thin and highly directional orthogonal beams is a much tighter one as compared to~\cite{karnin2013almost}. Of course, the additional assumption is that we allow for multi-beam concurrent transmission.
\begin{figure}
    \centering
    \includegraphics[width = 0.7\linewidth]{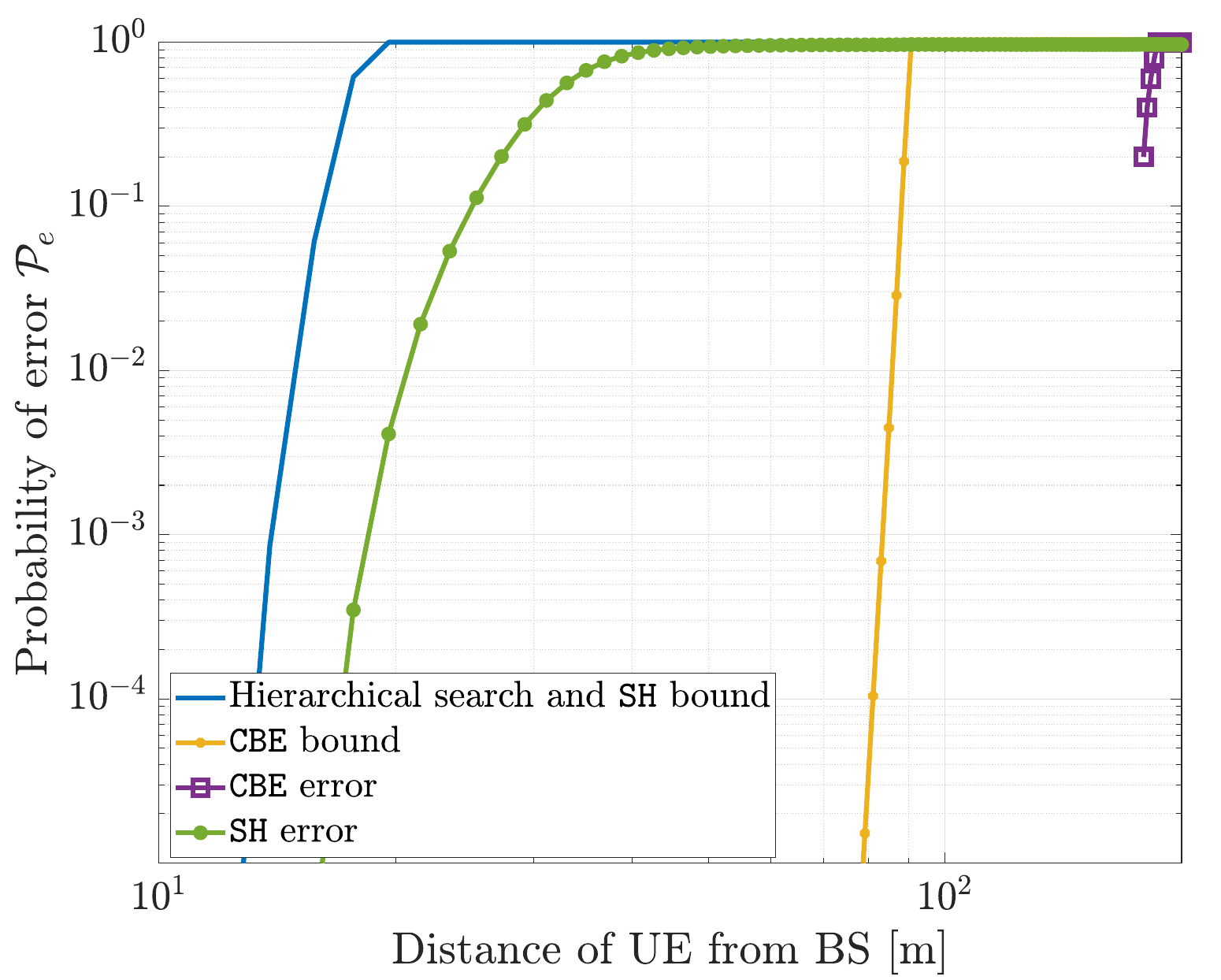}
    \caption{Probability of beam selection error and error bounds with respect to the UE distance.}
    \label{fig:compare_bound}
\end{figure}

This is confirmed in Fig.~\ref{fig:compare_bound} where we plot the bounds for the hierarchical search and for \texttt{CBE} with respect to the distance of the \ac{UE} from the \ac{BS}. For comparison, we also plot the actual error evaluated using extensive Monte-Carlo simulations.

\subsection{Impact of Sidelobe Gain}
\begin{figure}
    \centering
    \includegraphics[width = 0.5\textwidth]{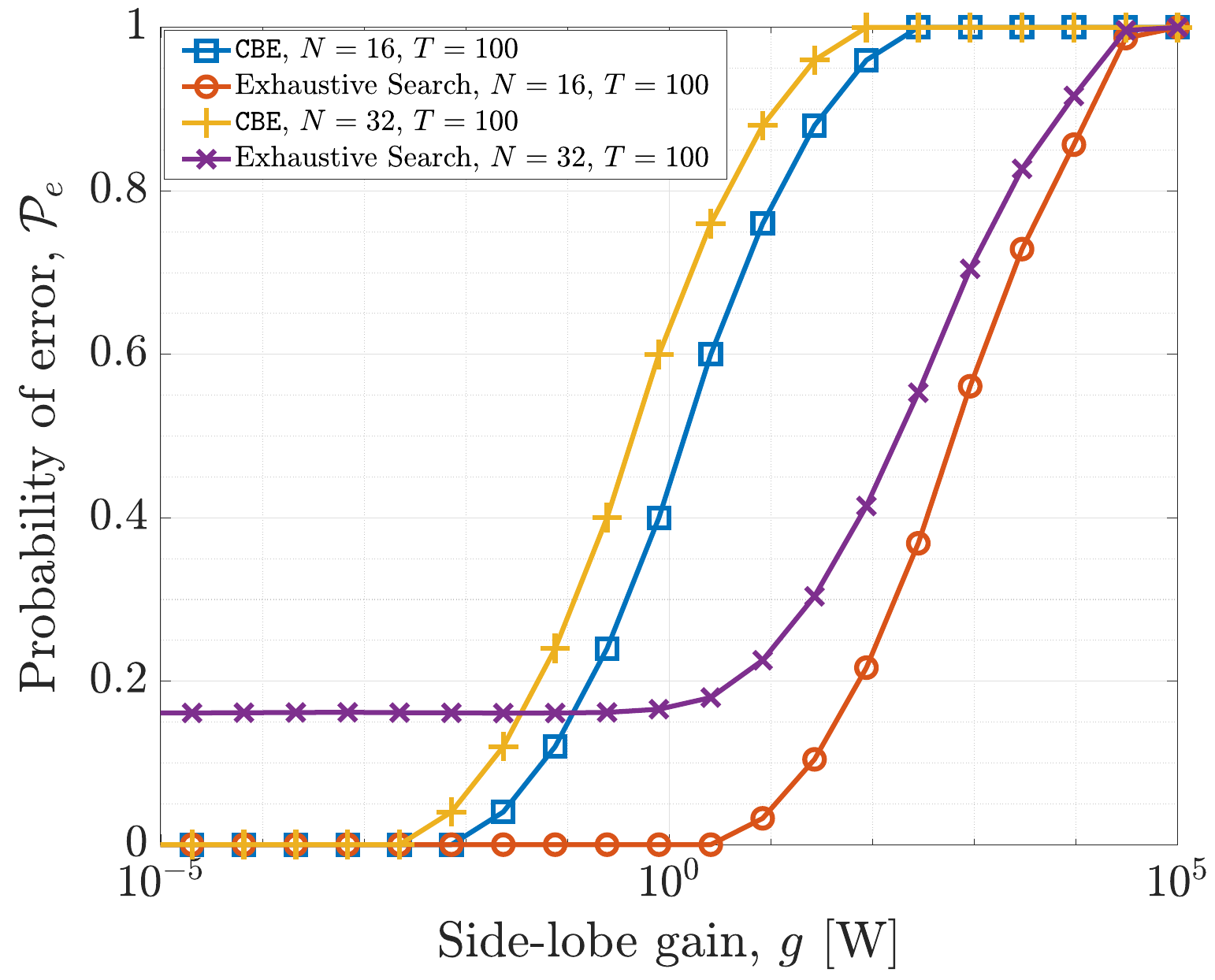}
    \caption{Impact of sidelobe gain $g$ on the probability of beam-selection error.}
    \label{fig:limits}
\end{figure}
{Although ultra-thin orthogonal beams are now possible~\cite{foo2015orthogonal}, often in practical systems, the sidelobes are non-negligible. This may lead to increased beam-selection errors. Fig.~\ref{fig:limits} confirms that the larger the value of $g$, the greater will be the probability of selecting the incorrect beam (both in the case of exhaustive search and \texttt{CBE}). For a fixed budget $T$, beyond a value of sidelobe gain, the exhaustive search approach results in a lower beam-selection error as compared to \texttt{CBE}. This is precisely because the multi-beam transmission in \texttt{CBE} results in the addition of the contribution from all the sidelobes and, accordingly, a higher variance of the received power, leading to larger errors. However, in modern mm-wave transceiver systems, the difference in the gains of the main and the sidelobe is over 15 dB~\cite{chae2018array}. Thus, in the case of highly directional antennas with limited sidelobe radiation, it is envisaged that the proposed \texttt{CBE} scheme will perform satisfactorily.}




\section{Best Beam Selection in an Abruptly Changing Environment}
\label{sec:BBSACE}
\subsection{The \texttt{SH} Algorithm}
The \texttt{SH} is a popular algorithm used for identifying the best arm in \ac{MAB} problems. It evolves as a sequence of rounds. The total number of rounds is $\log N$. In our context, for each round, each beam is allocated an equal number of measurement time slots for transmission. Within each round, the mean power of each beam is evaluated based on the allocated slots. Then, the top half of beams (i.e., those with the highest mean power) are identified and the remaining beams are eliminated. Then, in the next round, the framework allocates an equal number of slots to each of the surviving beams.
These steps are repeated until only one beam remains, which is considered the best beam based on the observed rewards. Thus, \texttt{SH} balances exploration and exploitation by gradually eliminating weaker beams and reallocating samples to the stronger beams. By allocating more measurement slots to beams with potentially higher power, it focuses exploration on the most promising options.

\subsection{\texttt{SH} with a Single Abrupt Change}
For this analysis, we make a minor change in the \texttt{SH} algorithm as compared to \cite{karnin2013almost} - in each episode, instead of consecutive sampling from the same beam, we sample the beams in a round-robin manner. Naturally, this increases the possibility of sampling the beams post a change event. Let us consider the estimated {mean power} of the $i-$th beam at the end of the $r-$th round in case it does not experience a change:
\begin{align}
    \hat{\mu}_i(r) = \frac{1}{n_r} \sum_{t= 1}^{n_r} R_i\left(\sum_{v =1}^{r-1}n_v + (i - |S_r|) + |S_r| t\right),
\end{align}
where $n_r = \frac{2^{r-1}T}{N\log  N }$ is the number of times each beam is sampled in the round $r$. We simplify the time indices since the reward values for beam ${\bf f}_i$ are i.i.d. as $R_i$ in case of no changes.  Thus, the {estimate of the mean power} of the beam ${\bf f}_i$ at the end of episode $r$ is given as
\begin{align}
    \hat{\mu}_{i}(r) = \frac{1}{n_r} \sum_{k= 1}^{n_r} R_{i}\left(k\right) = \frac{1}{n_r} \norm{{\bf R}_i(r)}^2. \nonumber 
\end{align}

Let the change occur in the reward distribution of the beam ${\bf f}_j$ in round $r_c$ and consider that the reward values for beam $j$ are i.i.d. as $R_j^- \sim \mathcal{N}\left(\mu_j^-, \sigma_j^-\right)$ before $t_c$ and as $R_j^+ \sim \mathcal{N}\left(\mu_j^+, \sigma_j^+\right)$ after $t_c$. Consider the case that the change results in the beam $j$ being the best beam for $t > t_c$, i.e., $\argmax \mu_i(t) = j$ for $t > t_c$. If $j \in S_{r_{\rm c}}$, its estimate of the mean power is
\begin{align}
    \hat{\mu}_j(r_c) = \frac{1}{n_{r_c}} \left[\sum_{l = 1}^{n'} R_j^-\left(l\right) + \sum_{m = 1}^{n_{r_c}-n'} R_j^+\left(m\right)\right], \nonumber 
\end{align}
where $n'$ is the slot in round $r_{\rm c}$ after which the change occurs. Let the conditional \ac{CDF} that the change occurs in any slot $n$ given $r_{\rm c}$ be given by $F_{n'}(n | n_{r_{\rm c}})$.

\subsection{Analysis for Round $r_{\rm c}$}
Recall that $S_{r_{\rm c}} = \frac{N}{2^{r_{\rm c} - 1}}$ beams enter the round $r_{\rm c}$ and each beam is played $n_{r_{\rm c}} = \frac{T}{|S_{r_{\rm c}}|\log N}$ times. The probability, $p_{i,j}(r_{\rm c})$ that the arm $j$ has a lower empirical mean than the arm $i \neq j$ after round $r_c$ is calculated in the following lemma.
\begin{lemma}
\label{lem:rc_bound}
Given that the beams ${\bf f}_i$ and ${\bf f}_j$ survive until the round $r_{\rm c}$ in which  ${\bf f}_j$ undergoes a change, the probability that the estimate of the mean power of ${\bf f}_j$ is lower than that of ${\bf f}_i$ after round $r_{\rm c}$ is
\begin{align}
    p_{i,j}(r_{\rm c}) \leq 1 - {F}_{n'}\left(n_i^* \mid n_{r_{\rm c}}\right) \left(1 -  \exp\left(- \frac{\Delta^2_{\min}T}{2N\log N\sigma_{\max}^2}\right)\right),
\end{align}
where $n_i^* = 
-\frac{n_{r_{\rm c}}\Delta_{i,j}^+}{\Delta_{\rm c}}$ and $F_{n'}(\cdot \mid n_{r_{\rm c}})$ is the conditional \ac{CDF} of the change time slot given that the change occurs in the round $r_{\rm c}$. Additionally, $\sigma^2_{\max} = 2\sigma^2 \mu_{\max}$.
\end{lemma}
\begin{IEEEproof}
Please see Appendix~\ref{app:rc_bound}.
\end{IEEEproof}

We note that the bound derived above has two parts - $1 - {F}_{n_i'}\left(n_i^* \mid n_{r_{\rm c}}\right)$ and ${F}_{n'}\left(n^* \mid n_{r_{\rm c}}\right)\exp\left(- \frac{\Delta^2_{\min}T}{2N\log N\sigma_{\max}^2}\right)$. Since we are interested in the event that ${\bf f}_j$ survives the round, the sum of these terms needs to be less than or equal to 1 for the bound to be meaningful. Based on the difference of the means between the arms ${\bf f}_i$ and ${\bf f}_j$ before and after the change, the following four cases arise.
\begin{enumerate}
    \item $\Delta_{i,j}^+ > 0 $ and $\Delta_{\rm c} > 0$, i.e., the beam ${\bf f}_i$ is always superior to the beam ${\bf f}_j$. In this case, $p_{i,j}(r_{\rm c})$ is trivially upper bounded by 1.
    \item $\Delta_{i,j}^+ < 0 < \Delta_{\rm c}$ and $|\Delta_{i,j}^+| > |\Delta_{\rm c}|$, i.e., the beam ${\bf f}_i$ is always inferior to the beam ${\bf f}_j$. In this case, $F_{n'}(n_i^* \mid n_{r_{\rm c}}) = 1$ and thus, $p_{i,j}(r_{\rm c})$ is exponentially bounded.
    \item $\Delta_{i,j}^+ < 0 < \Delta_{\rm c}$ and $|\Delta_{i,j}^+| < |\Delta_{\rm c}|$, i.e., the beam ${\bf f}_i$ is superior to ${\bf f}_j$ before the change and it becomes inferior to the beam ${\bf f}_j$ after the change. Here for a change at slot $n' \leq n_i^*$, $p_{i,j}(r_{\rm c})$ is exponentially bound, while for $n' > n_i^*$ it is trivially bounded by 1. Hence, in this case, the earlier the change, the higher the change that ${\bf f}_j$ survives with respect to ${\bf f}_i$.
    \item $\Delta_{\rm c}< 0 < \Delta_{i,j}^+$, i.e., ${\bf f}_i$ is inferior to ${\bf f}_j$ before the change and it becomes superior to ${\bf f}_j$ after the change. Contrary to the previous case, here, for a change at slot $n' \leq n_i^*$, $p_{i,j}(r_{\rm c})$ is bounded by 1, while for $n' > n_i^*$ it is exponentially bounded. Hence, in this case, the later the change, the higher the chance that ${\bf f}_j$ survives with respect to ${\bf f}_i$.
\end{enumerate}
Out of the above, only cases 2 and 3 are of interest to us since we assume that after the change ${\bf f}_j$ becomes the best beam. Let the change occur in the $K-$th best beam. Then, the following result bound its probability of elimination in the round $r_{\rm c}$.
\begin{lemma}
\label{lem:changeround}
    The probability that the $K-$th arm is eliminated in round $r_{\rm c}$ is upper bounded by
    \begin{align}
    p_{K}(r_{\rm c})& \leq 2\left[1 - 
    {F}_{n'}\left(n_{\max} \right) \left(1 -  \exp\left(- \frac{\Delta_{\min}^2}{2\sigma_{\max}^2}\right)\right)\right], \nonumber 
\end{align}
where $n_{\max} =  \frac{n_{r_{\rm c}} \Delta_{\min}}{\Delta_{\rm c}}$.
\end{lemma}
\begin{IEEEproof}
Let $N_{r_{\rm c}}$ denote the number of arms that have a higher estimated mean than the $K-$th arm in the round $r_{\rm c}$. Then,
\begin{align}
        \mathbb{E}\left[N_{r_{\rm c}}\right] &= \sum_{{\bf f}_i \in \mathcal{S}_{r_{\rm c}}} \mathbb{P}\left(\hat{\mu}_i(r_{\rm c}) > \hat{\mu}_K(r_{\rm c})\right)  \nonumber  \\ 
        &\leq \sum_{{\bf f}_i \in \mathcal{S}_{r_{\rm c}}} 1 - 
    {F}_{n'}\left(n_i^* \mid n_{r_{\rm c}} \right) \left(1 -  \exp\left(- \frac{{\Delta_{\min}}^2}{2\sigma_{\max}^2}\right)\right) \nonumber \\
    &\leq |\mathcal{S}_{r_{\rm c}}| \left[1 - 
    {F}_{n'}\left(n_{\max} \right) \left(1 -  \exp\left(- \frac{\Delta_{\min}^2}{2\sigma_{\max}^2}\right)\right)\right]. \label{eq:av_rc}
\end{align}
Now, from Markov's inequality, we have
\begin{align}
    \mathbb{P}\left(N_{r_{\rm c}} \geq \frac{|\mathcal{S}_{r_{\rm c}}|}{2}\right) &\leq \frac{2\mathbb{E}\left[N_{r_{\rm c}}\right]}{|\mathcal{S}_{r_{\rm c}}|}. \nonumber
\end{align}
Substituting \eqref{eq:av_rc} in the above completes the proof.
\end{IEEEproof}

\begin{example}
    In case the exact change slot is uniformly distributed in the round $r_{\rm c}$, then we have $F_{n'}(n_{\max}) = \frac{n_{\max}}{n_{r_{\rm c}}} = \frac{\Delta_{\min}}{\Delta_{\rm C}}$. Accordingly, $p_{K}(r_{\rm c})$ is upper bounded as
    \begin{align}
        p_{K}(r_c) \leq 2\left(1 - \frac{\Delta_{\min}}{\Delta_{\rm c}}\left(1 -  \exp\left(- \frac{\Delta_{\min}^2}{2\sigma_{\max}^2}\right)\right)\right).
    \end{align}
\end{example}

\begin{figure}
    \centering
    \includegraphics[width = 0.7\linewidth]{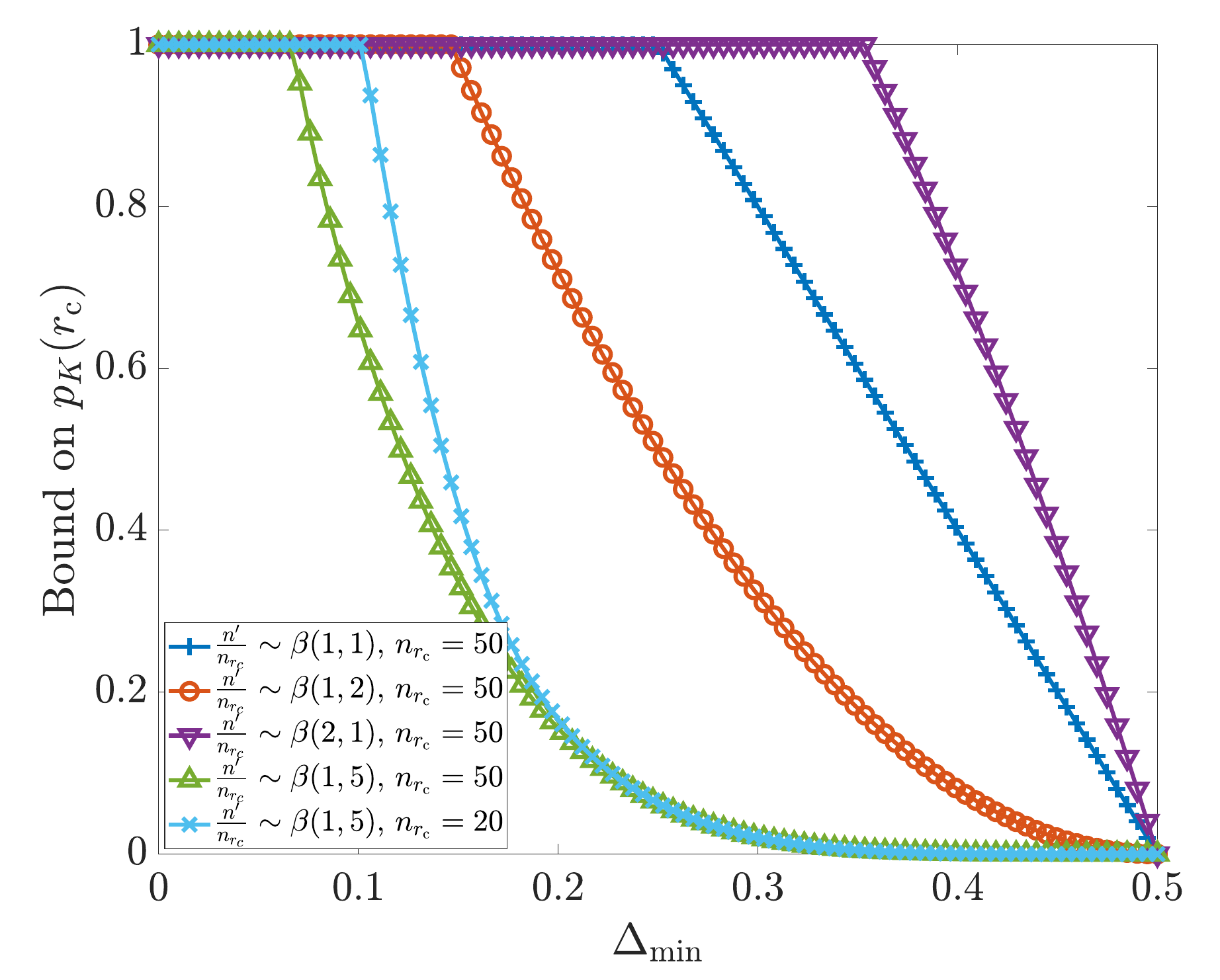}
    \caption{Bound on $p_{K}(r_{\rm c})$ with respect to $\Delta_{\rm min}$ for different change distributions.}
    \label{fig:change_bound}
\end{figure}
For a given value of $r_{\rm c}$ (equivalently $n_{r_{\rm c}}$) the exact location of the change is governed by its distribution. In this work we do not make any assumptions on the same, and hence, a beta distribution is appropriate to model its location~\cite{gupta2004handbook}. First we note from Fig.~\ref{fig:change_bound} that higher the magnitude of change $\Delta_{\rm min}$, the lower will be the bound on $p_{K}(r_{\rm c})$. More importantly, in case the changes occur earlier in the change round, i.e., the beta distribution is skewed to the left, the probability of elimination of ${\bf f}_K$ is limited. In particular, we have the following important result.

\begin{corollary}
    For $r_{\rm c} \leq r^*$,  
     \begin{align}
        p_{K}(r_c) \leq 2K\exp\left(- \frac{\Delta_{\min}^2}{2\sigma_{\max}^2}\right).
    \end{align}
\end{corollary}
\begin{IEEEproof}
This follows from Lemma~\ref{lem:changeround} by recognizing that for all beams ${\bf f}_i$ which are inferior to ${\bf f}_K$, we have $\Delta_{i,K}^+ < 0 < \Delta_{\rm c}$ and $|\Delta_{i,K}^+| > |\Delta_{\rm c}|$. Hence, $F_{n'}(n_i^*) = 1 \forall i > K$. Now, for ${\bf f}_K$ to be eliminated, it has to be in the bottom half of the estimated beams in the $r_{\rm c}-$th round, at least $\frac{|\mathcal{S}_{r_{\rm c}|}}{2} - K$ inferior beams should have a higher estimate than ${\bf f}_K$. Recall that the number of beams in the $r_{\rm c}-$th round which are inferior to ${\bf f}_K$ is ${|\mathcal{S}_{r_{\rm c}|}} - K$. Hence,
    \begin{align}
         \mathbb{P}\left(N_{r_{\rm c}} \geq \frac{|\mathcal{S}_{r_{\rm c}}|}{2} \mid r_{\rm c} \leq r^*\right) &\leq \frac{{|\mathcal{S}_{r_{\rm c}|}} - K}{{\frac{|\mathcal{S}_{r_{\rm c}|}}{2} - K}} \exp\left(- \frac{\Delta_{\min}^2}{2\sigma_{\max}^2}\right)\nonumber \\
         &\leq 2K\exp\left(- \frac{\Delta_{\min}^2}{2\sigma_{\max}^2}\right). \nonumber 
    \end{align}
\end{IEEEproof}

Next, we characterize the probability of eliminating the $K$-th arm $(1 \leq K \leq N)$ in two distinct segments.
\\

{\bf Early change - $r_{\rm c} \leq r^*$}
\begin{lemma}
\label{lem:ec}
    Conditioned on the change occurring within the first $r^* = \log \frac{N}{2K} + 1$ rounds, the probability that the best arm is eliminated is upper bounded as
\begin{align}
    \mathcal{P}^{\rm SH}_{\rm C}\left(T \mid r_{\rm c} \leq r^* \right) \leq 2 \left(\log  N + K - 1\right)\exp\left(-\frac{1}{2}\frac{\Delta^2_{\min}T}{N\log N }\right). \nonumber 
\end{align}
\end{lemma}
\begin{IEEEproof}
Please see Appendix~\ref{sec:proof_ec}.
\end{IEEEproof}


{\bf Late change - $r_{\rm c} > \log \frac{N}{2K}$}

\begin{lemma}
\label{lem:lc}
    If the change occurs after the first $\log \frac{N}{2K}$ rounds, the probability that the best arm is eliminated is upper bounded as
    \begin{align}
        \mathcal{P}^{\rm SH}_{\rm C}\left(T \mid r_{\rm c} > r^* \right) \leq  &\mathcal{T}_1(r_{\rm c}) + \nonumber \\
        &2\log 2NK \exp\left(-\frac{1}{2} \Delta_{\min}^2 \frac{T}{N \log N}\right). \nonumber 
    \end{align}
    where $\mathcal{T}_1(r_{\rm c}) = \mathbb{E}\left[r_{\rm c} - r^* , r^* \leq r_{\rm c} \leq \log N\right]$.
\end{lemma}
\begin{IEEEproof}
Please see Appendix~\ref{sec:proof_lc}
\end{IEEEproof}
In the case of late change, the bound has an exponential term and a term that depends on the distribution of the change slot location. Thus, in cases of late change, \texttt{SH} does not achieve an exponential upper bound. 

\begin{theorem}
    In case of a single abrupt change in the mean power of ${\bf f}_K$ at time $0 \leq t_{\rm c} \leq T$, the bound on the beam selection error is given by
    \begin{align}
        \mathcal{P}^{\rm SH}_{\rm C} =  &\mathcal{P}^{\rm SH}_{\rm C}\left(T \mid r_{\rm c} \leq r^* \right) \mathbb{P}\left( r_{\rm c} \leq r^*\right) +  \nonumber \\
        &\mathcal{P}^{\rm SH}_{\rm C}\left(T \mid r_{\rm c} > r^* \right)\mathbb{P}\left( r_{\rm c} > r^*\right) \nonumber \\
       & \leq \mathcal{T}_1(r_{\rm c}) + 
        2 \left(2\log  N + K - 1\right)\exp\left(-\frac{1}{2}\frac{\Delta^2_{\min}T}{N\log N }\right). \nonumber 
    \end{align}
\end{theorem}

\begin{corollary}{\bf [No Change]}
In case of no change, the performance of \texttt{SH} is exponentially bounded as $\log N \exp\left(-\frac{1}{2}\frac{\Delta^2_{\min}T}{N\log N }\right)$, which is of the form given in~\cite{karnin2013almost}.
\end{corollary}

\begin{algorithm}
\caption{\texttt{K-SHES}}\label{alg:kshes}
\begin{algorithmic}[1]
\State {\bf Input:  $[{\bf F}]$}.
\State $\mathcal{S}_1 = [{\bf F}]$.
\State Calculate $r^* = \log \frac{N}{2K}$.
\For{$1 \leq  r < r^*$}
    \State Calculate $n_{r} = \frac{T}{\log N 2^{r- 1}}$.
    \For{All ${\bf f}_{i} \in S_r$}
        \State Measure received power $R_i$ in slots $k|\mathcal{S}_r + i|$ slots for $k = 0, 1, \ldots, n_{r} - 1$.
        \State Evaluate $||{\bf R}_i||^2$.
    \EndFor
    \State Rank the beams of $S_r$ in decreasing order of $||{\bf R}_i||^2$.
    \State Identify the set $\mathcal{S}'_r$ bottom $\frac{|\mathcal{S}_r|}{2}$ arms.
    \State $\mathcal{S}_{r+ 1} = \mathcal{S}_r \backslash \mathcal{S}'_r$.
\EndFor
    \State From the remaining $\mathcal{S}_{r^* + 1}$ arms, identify the best beam using equal allocation.
    \State {\bf Return:} Optimal beam ${\bf f}_j$.
\end{algorithmic}
\end{algorithm}

\section{Hybrid Policy for known $K$}
\label{sec:HPKK}
Next, consider the case when the change is restricted to the top $K$ beams of the system. This is typical for cases when the optimal beam is blocked initially. The beam-selection procedure recognizes an adjacent beam to the optimal beam as the best one for service initially. This is mainly due to the correlation among the beams, directional transmissions, and limited multipath in \ac{mm-wave}. However, in the event that the optimal beam abruptly transitions into a line-of-sight state during the beam-selection procedure, the algorithm must adapt and report only the optimal beam. In this regard, we propose \texttt{K-SHES}, which exploits the knowledge of $K$ to tune the \texttt{SH} appropriately. The steps of \texttt{K-SHES} are presented in Algorithm \ref{alg:kshes}. For a given value of $K$, we calculate $r^* = \frac{N}{2K}$. Until the round $r^*$, \texttt{K-SHES} employs the classical \texttt{SH} algorithm, i.e., until the $2K$ arms are left. Once $2K$ beams are left, the algorithm does not further eliminate beams. After $r^*$, the remaining $\frac{N}{2^{r^* - 1}}$ beams are sampled in a round-robin manner, and the best beam is determined at $T$ based on received power in the slots after $r^*$.
\begin{theorem}
    The beam selection error for the \texttt{K-SHES} algorithm is given by
    \begin{align}
        &\mathcal{P}^{\rm K-SHES}_{C}(T) \leq  \mathcal{T}_{\rm K-SHES}\exp\left(-\frac{1}{2}\frac{\Delta^2_{\rm min}T}{2N\log N \sigma^2_{\max}}\right) + \nonumber \\
        &\sum_{i = 1}^{K-1} 1 - F_{t_{\rm c}} (t_i \mid r^*)\left(1 - \exp\left(-\frac{1}{2}\frac{\Delta^2_{\rm min}}{N\log  N \sigma^2_{\rm max}}\right)\right), \nonumber 
    \end{align}
where $\mathcal{T}_{K-SHES} = \log  \frac{N^2}{2K} + K \left(2\log  (2K) + 1\right)$. If the change occurs in the first $T\left[\frac{\log(N/2K)}{\log  N}\left(1 - \frac{1}{2K}\right)  + \frac{1}{2K}\right]$ time slots with probability 1, then the beam selection error is exponentially bounded as
    \begin{align}
         \mathcal{P}^{\rm K-SHES}_{C}(T) \leq & 2\left(2\log  N + 2K  - 1\right) \cdot \nonumber \\
        &\exp\left(-\frac{1}{2}\frac{\Delta^2_{\rm min}T}{2N\log N \sigma^2_{\max}}\right)
    \end{align}
\end{theorem}
\begin{proof}
Similar to the \texttt{SH} case, the upper bound on the error with \texttt{K-SHES} can be derived as a sum for the early and the late change cases. For both the early change and the late change cases, the analysis remains the same until the round $r^*$. Beyond $r^*$, due to a single round, the probability of beam-selection error is given by the union bound over the remaining $2K$ arms, similar to the exhaustive search case.
    \begin{align}
        \mathbb{P}\left(\mathcal{E}_K([r^*,\log N]) \mid r_{\rm C} \leq r^*\right) \leq 2K\exp\left(-\frac{1}{2}\frac{\Delta^2_{\rm min}}{N\log  N \sigma^2_{\rm max}}\right) \nonumber 
    \end{align}
    While for the late change case, the analysis follows similarly to Lemma~\ref{lem:changeround},
        \begin{align}
        \mathbb{P}\left(\mathcal{E}_K([r^*,\log N]) \mid r_{\rm C} > r^*\right) \leq \sum_{i = 1}^{K-1} 1 - F_{t_{\rm c}} (t_i \mid r^*) \nonumber \\
        \left(1 - \exp\left(-\frac{1}{2}\frac{\Delta^2_{\rm min}}{N\log  N \sigma^2_{\rm max}}\right)\right) \nonumber 
    \end{align}
For the late change case, if the change occurs early enough so as to safeguard against the best arm prior to the change, $\texttt{K-SHES}$ results in an exponential bound.
\end{proof}

\section{Numerical Results and Discussion}
\label{sec:NRD}
\subsection{Performance Comparison}
Fig.~\ref{fig:compare_bound3} shows that \texttt{K-SHES} outperforms \texttt{SH} and the exhaustive search algorithms. However, due to the no elimination in the case of \texttt{K-SHES}, it suffers from a higher probability of error as compared to \texttt{SH} in cases where the changes occur early. In case of an early change, \texttt{SH} and \texttt{K-SHES} perform equally until $r^*$. However, beyond $r^*$, due to no further changes, \texttt{SH} performs better due to sequentially eliminating suboptimal beams. However, since \texttt{K-SHES} does not eliminate beams beyond $r^*$, it results in a higher beam selection error. This is elaborated in Fig.~\ref{fig:compare_bound2}.

\begin{figure}
    \centering
    \includegraphics[width = 0.7\linewidth]{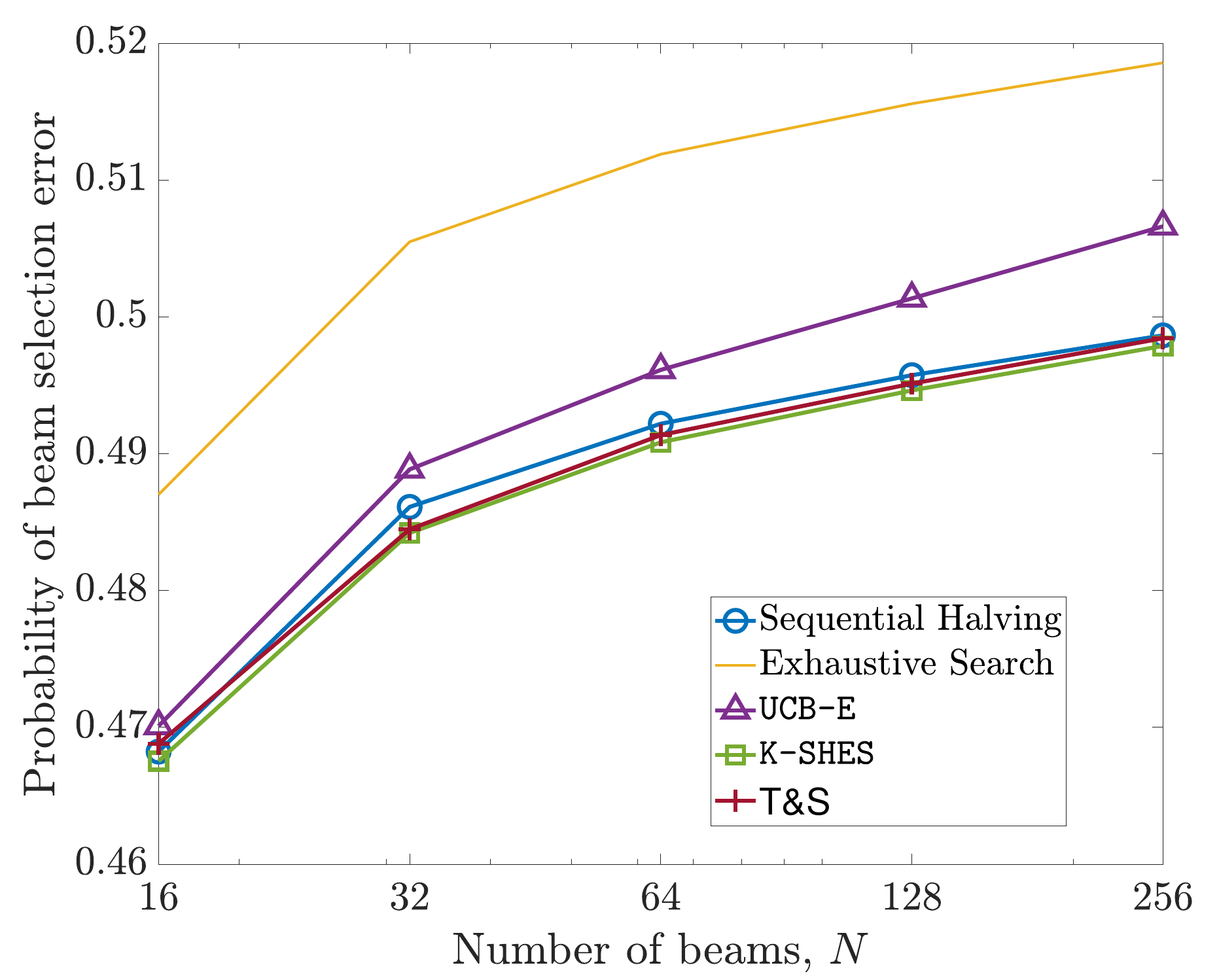}
    \caption{Comparison in terms of beam selection error. Here $T = 1024 = 2^{10}$ and $K$ belongs to the top 20 percent of the beams.}
    \label{fig:compare_bound3}
\end{figure}
\begin{figure}
    \centering
    \includegraphics[width = 0.7\linewidth]{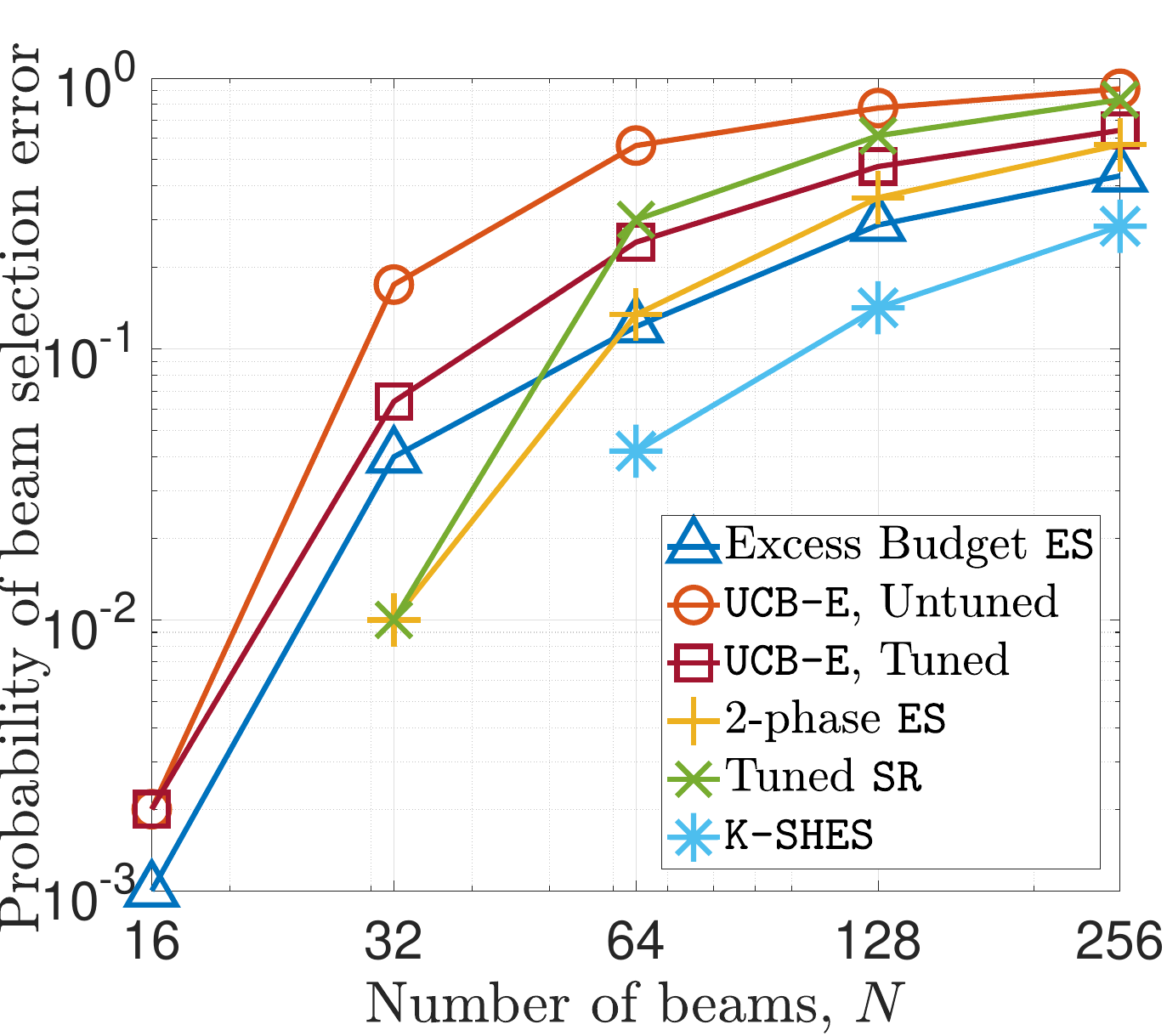}
    \caption{Comparison of different best arm identification algorithms with prior knowledge of $K$. $T = 25600$ slots.}
    \label{fig:compK}
\end{figure}
 {Next, we compare \texttt{K-SHES} with other algorithms that consider prior knowledge of $K$. It is important to highlight that the player does not have prior knowledge of the top $K$ beams in the system. It knows the maximum index of the beam that can undergo a change but does not know which these beams are. To the best of our knowledge, there are no known beam-selection algorithms that consider possible optimum beam identity changes during an \ac{SSB} burst, either with or without information about $K$. Accordingly, we create modified versions of classical best-arm identification algorithms for comparison. In particular, we compare the performance of \texttt{K-SHES} with
     \begin{itemize}
    \item {\bf 2-phase Exhaustive Search (ES):} In this algorithm, a fraction of the fixed budget of time slots is reserved to identify the best $K$ arms, and it is followed by an exhaustive search over the selection $K$ arms. Note that the change can still occur at any time slot uniformly selected in the fixed budget.
    \item {\bf Excess budget ES:} This is a modified version of the 2-phase ES, where we provision an extra time budget for identifying the best $K$ arms before proceeding with ES.
    \item {\bf Tuned UCB-E:} Here, based on the knowledge of $K$, we tune the exploration parameter of \texttt{UCB-E} so as to quickly determine the bottom $N-K$ arms, beyond which we employ \texttt{UCB-E} only on the top-K arms.
    \item {\bf Tuned SR:} Here we employ the classical successive rejects algorithm (see \cite{audibert2010best}) until $K$ arms are left beyond which we perform an exhaustive search.
\end{itemize}
Fig.~\ref{fig:compK} shows that \texttt{K-SHES} achieves the lowest beam selection error among its competitors. precisely due to the quickest elimination of the $N-K$ arms and then restricting the search to only the top arms.}

\begin{figure}
    \centering
    \includegraphics[width = 0.7\linewidth]{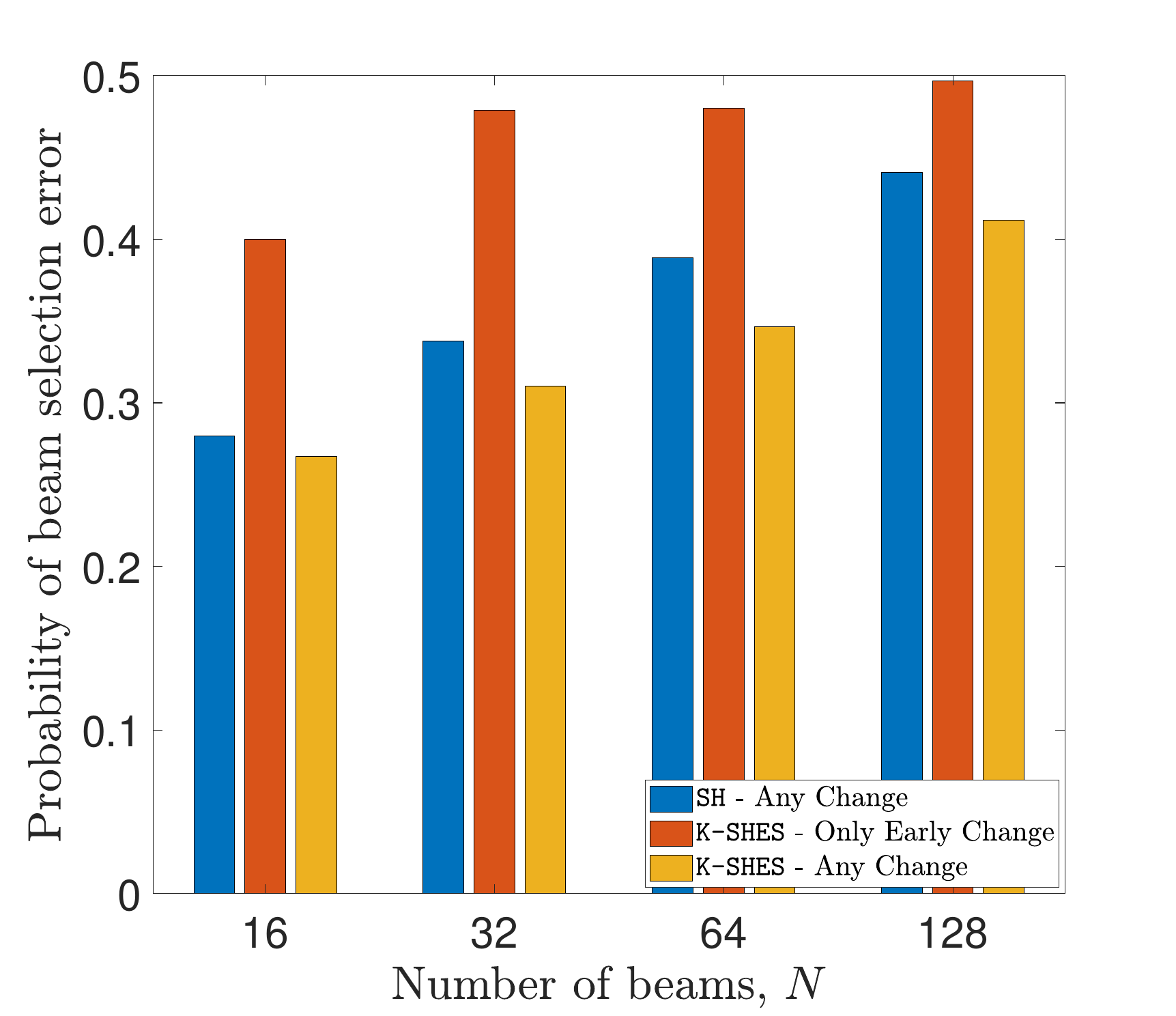}
    \caption{Impact of the location of change. Here $T = 4096 = 2^{12}$ and $K$ belongs to the top 20 percent of the beams.}
    \label{fig:compare_bound2}
\end{figure}

\subsection{Tuning of \texttt{K-SHES}}
 \begin{figure}
    \centering
    \includegraphics[width = 0.7\linewidth]{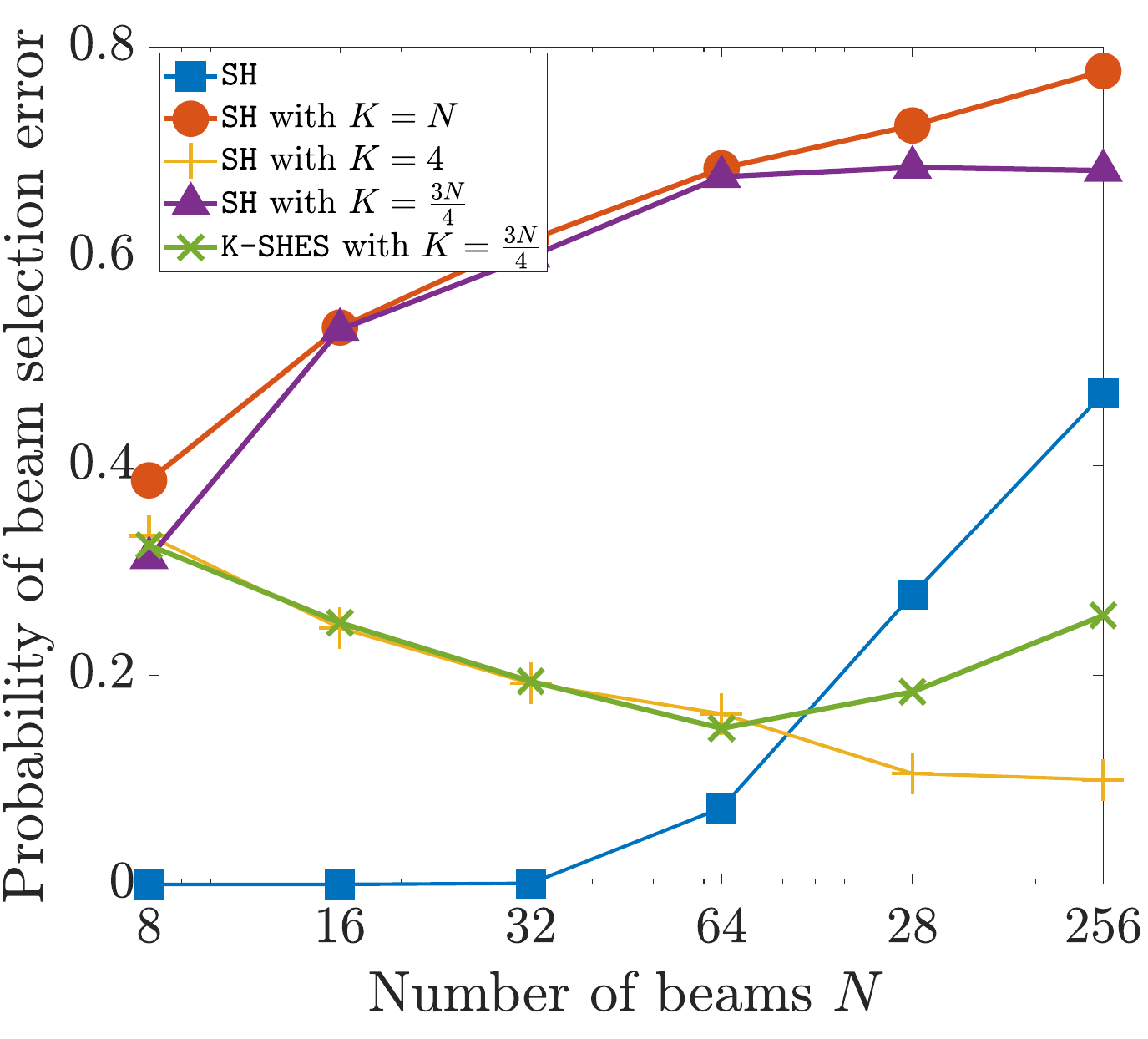}
    \caption{Impact of the maximum change index on the probability of beam selection error.}
    \label{fig:errork2}
\end{figure}

\begin{figure}
    \centering
    \includegraphics[width = 0.7\linewidth]{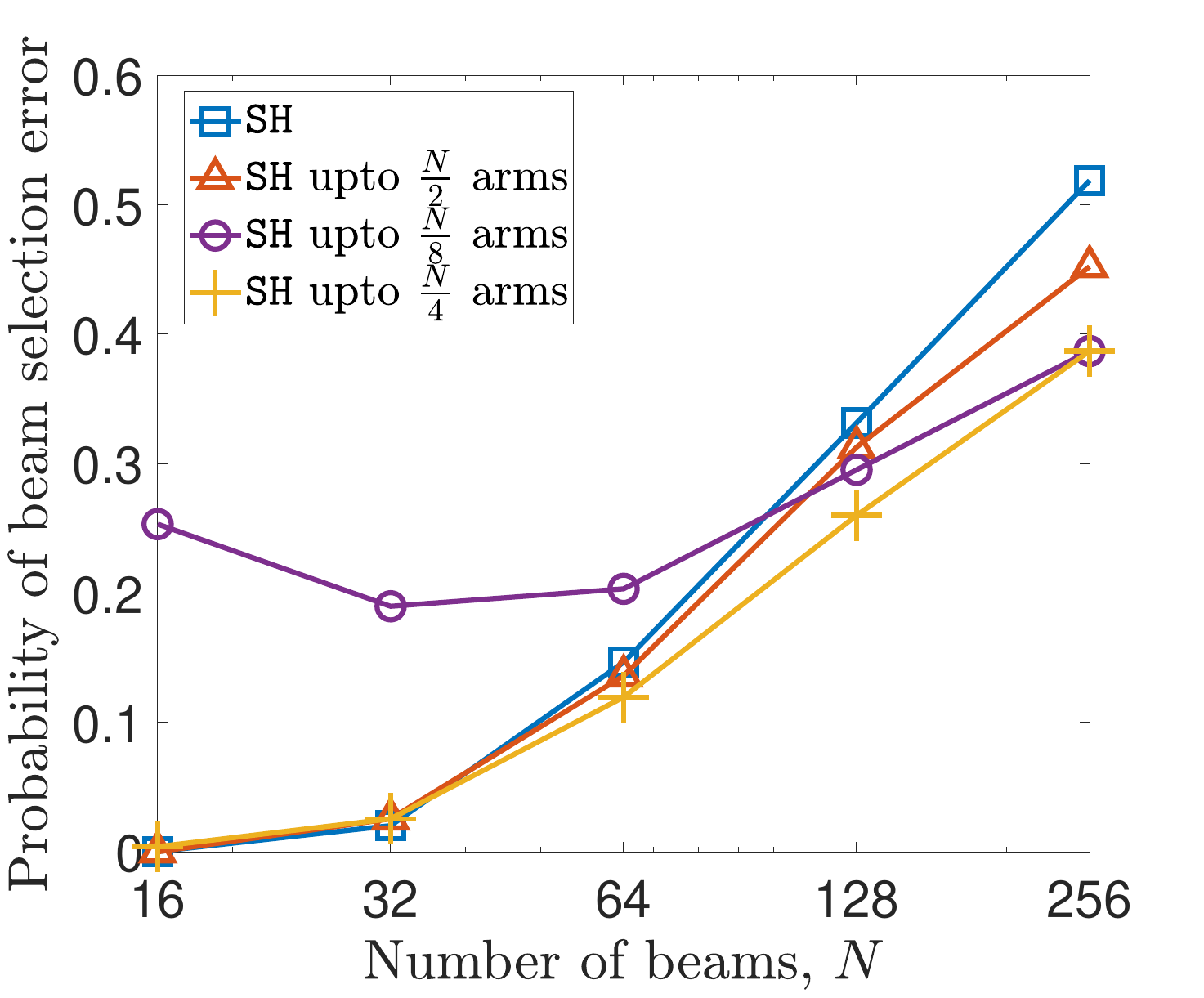}
    \caption{Impact of selecting tuning the \texttt{K-SHES} algorithm for a maximum change index of $K = \frac{N}{8}$.}
    \label{fig:review1}
\end{figure}
{In Fig.~\ref{fig:errork2} we plot the probability of mean selection error for different values of maximum beam index that undergoes a change at a uniformly selected change point. Thus, $K = N$ refers to the case where any beam can undergo a change. For some values of $K$, the probability of beam selection error first decreases with $N$. Albeit counter-intuitive, the reason for this observation is the fact that a change in the average downlink power of a sub-optimal beam to make it optimal also increases the ability of the algorithm to differentiate it from other sub-optimal beams. Hence, for lower values of $K$, an increased sub-optimality gap between the new optimal arm and the other sub-optimal arms.}

{Next, we study the tuning of the algorithm \texttt{K-SHES}. In Fig.~\ref{fig:review1}, we plot the probability of beam selection error by deploying \texttt{SH} for different numbers of phases. We set $K = \frac{N}{8}$, i.e., the change can occur only in the best $N/8$ beams. Then, as different candidates for comparison, we employ \texttt{SH} for $N$ (pure sequential halving), $N/2$, $N/4$, and $N/8$ phases before switching to exhaustive search on the remaining arms. We note that an early switch to the exhaustive search phase (e.g., \texttt{SH} up to $N/2$) results in a large number of time slots for exhaustively searching for the best beam among a large number ($N/2$) of remaining beams. This leads to a larger beam selection error. On the contrary, carrying out \texttt{SH} for a longer time (e.g., up to $N/8$ remaining beams) may result in the elimination of the beam which later becomes optimal after the change. This leads to a larger beam selection error as well. As hypothesized in our work, the appropriate balance between these two phenomena is achieved when we employ \texttt{K-SHES}, i.e., sequential halving until $2K$ arms remain in contention, followed by an exhaustive search.}

\subsection{Communication-Sensing Trade-off}
Next, let us study the efficacy of \texttt{K-SHES} from the perspective of a wireless communication system and the trade-offs arising from the same. Let the communication scheme be partitioned into beam refinement and downlink data transmission phases as shown in Fig.~\ref{fig:frames}. 
\begin{figure}
    \centering
    \includegraphics[width = 0.7\linewidth]{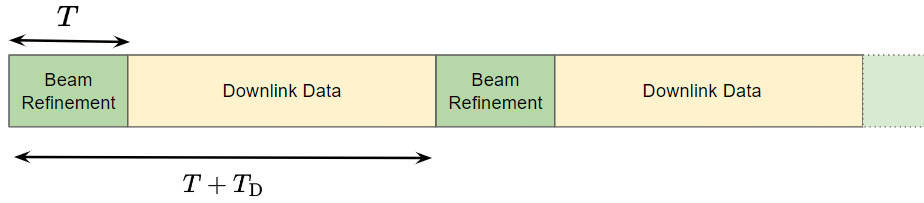}
    \caption{Illustration of the communication scheme.}
    \label{fig:frames}
\end{figure}
The beam alignment phase of duration $T$ is mapped to \texttt{K-SHES} developed in this paper. The data transmission phase is of duration $T_{\rm D}$. In the case of a larger $N$, each beam can be made highly directional, which leads to higher radiated power. However, a larger $N$ results in a higher beam selection error and reduced communication performance. In addition, for a fixed frame length $T_{\rm tot}$, if a higher number of time slots is allotted to beam refinement, then fewer slots remain for data communication, which may degrade the communication performance. On the contrary, fewer slots are reserved for beam refinement, which leads to a higher beam selection error and, accordingly, poor communication even with a large number of data transmission slots.

We assume that the user is stationary and is present at 100 m from the access point. The blockage condition can intermittently change uniformly within a frame. We assume a bandwidth of $1$ GHz and a transmit power of $40$ dBm. The impact of interference is ignored. The frame duration consists of $35072$ slots. For a given number of beam $N$, we assume that the directivity gain per beam is $\frac{2\pi}{N}$ and accordingly, the downlink data rate is given by $\left(1 - \mathcal{P}_{\rm e}\right) \frac{T_{\rm D}}{T  + T_{\rm D}} W \log_2 \left(1 + \xi_0  \right)$, where $\xi_0$ is the reference \ac{SNR} without the directivity gain. Here we have assumed that the side lobes have negligible power, and hence, the received power in case of a beam misalignment is 0.

\begin{figure}
    \centering
    \includegraphics[width = 0.7\linewidth]{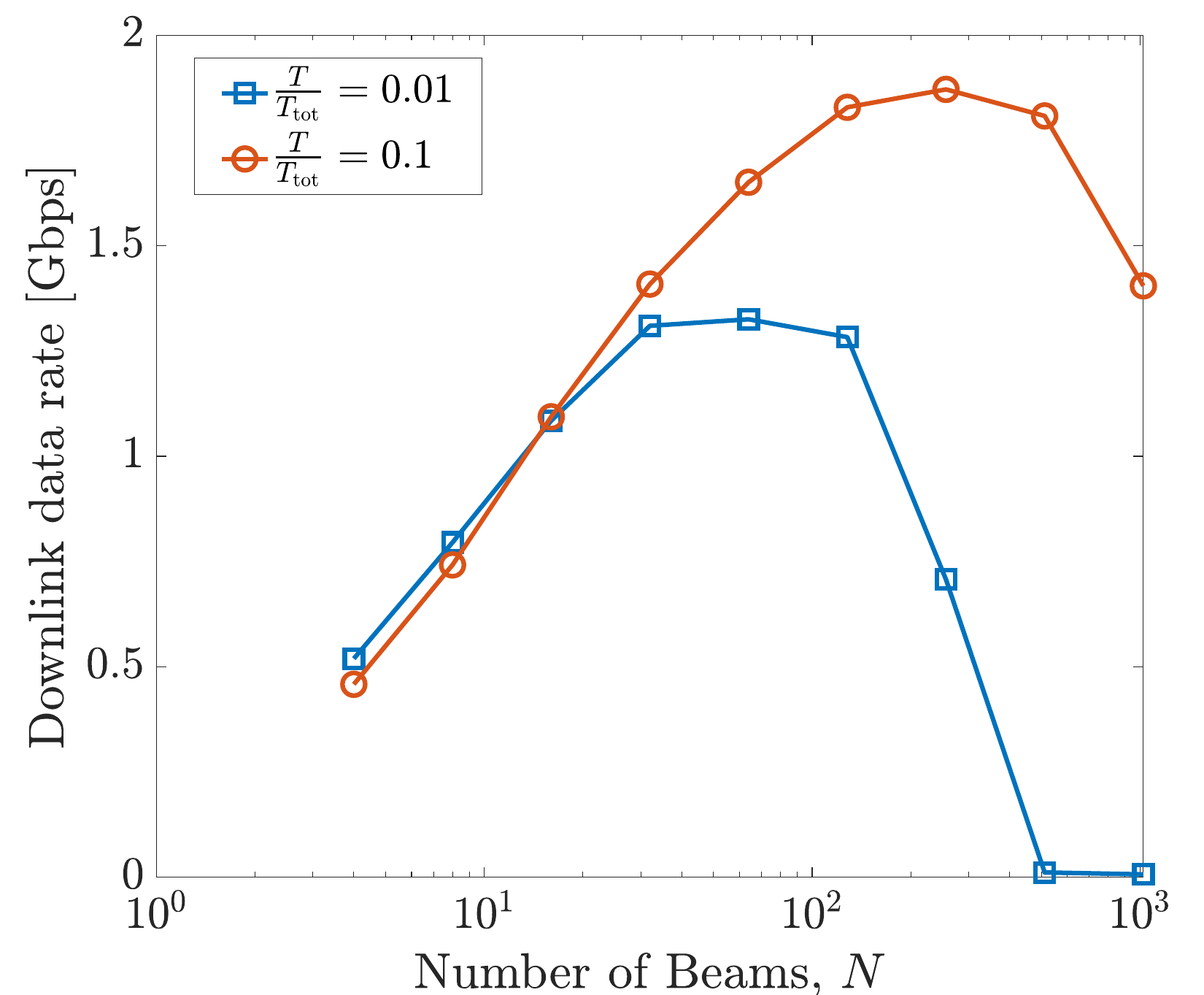}
    \caption{Downlink data rate with respect to number of beams for different fraction of frames allotted to beam alignment.}
    \label{fig:result1}
\end{figure}

Fig.~\ref{fig:result1} shows that for a chosen $T$, there exists an optimal $N$. In the case of 1\% of temporal resources allotted to beam alignment, the optimal beam number is 64, while for a higher number of resources allotted for beam alignment (10\%), the optimal $N$ increases to 128. Thus, for a larger beam alignment budget, a larger $N$ can be employed to maximize the data rate. However, for low beam dictionary size, e.g., $N = 16$, $\frac{T}{T_{\rm tot}}  = 0.1$ is sufficient to achieve the best possible beam selection efficacy, and increasing resources further for beam alignment simply reduces the data rate.
\begin{figure}
    \centering
    \includegraphics[width = 0.7\linewidth]{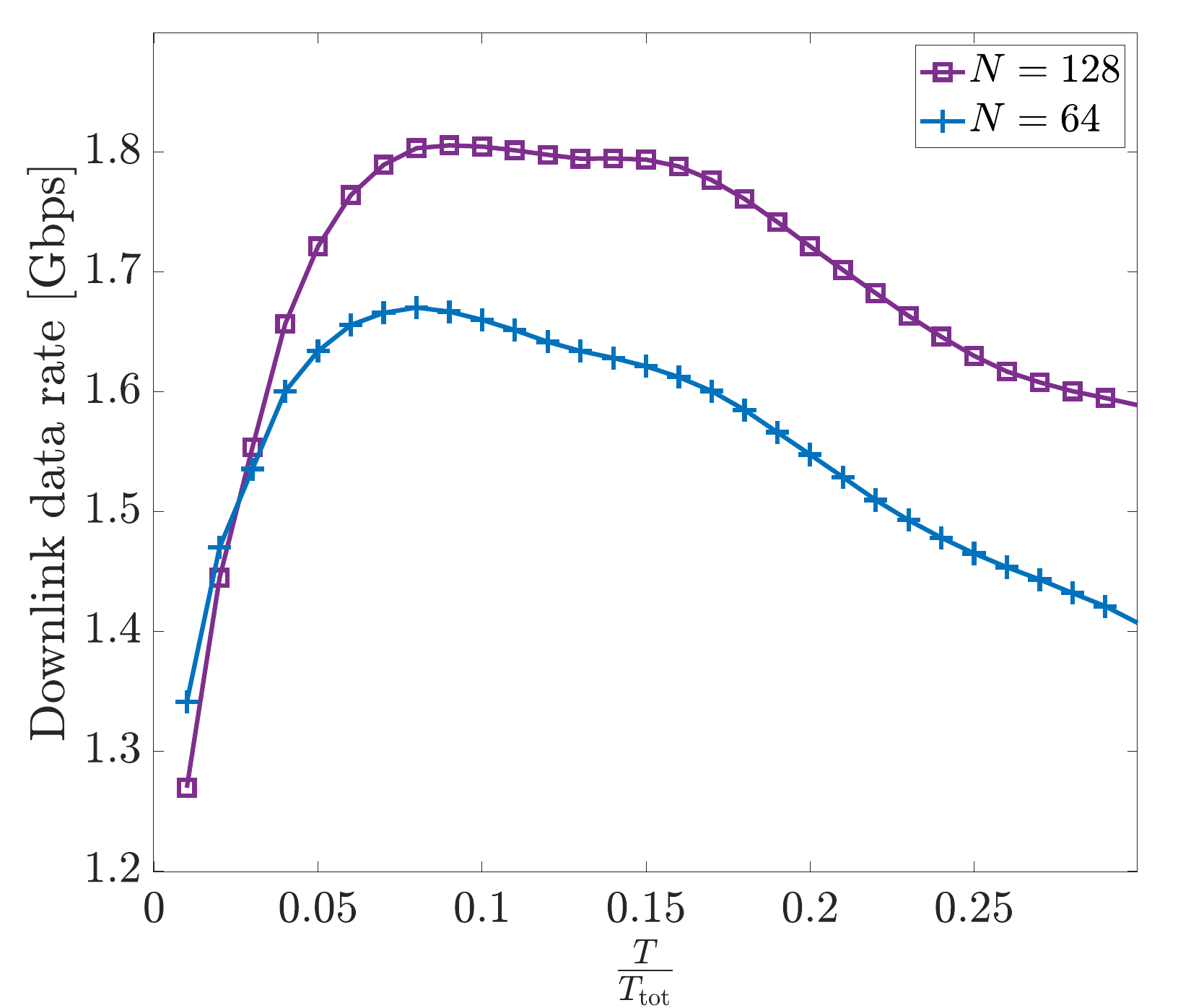}
    \caption{Data rate with respect to the fraction of resources allotted to beam refinement phase for different sizes of beam dictionary.}
    \label{fig:result2}
\end{figure}
Similarly, Fig.~\ref{fig:result2} shows that for a given $N$, there exists optimal partitioning of the temporal resources between beam alignment and data communication phases. For very stringent beam alignment deadline, e.g., $\frac{T}{T_{\rm tot}} = 0.1$, a lower $N$ is a better choice due to the low beam selection error. However, as the beam alignment budget increases, a higher $N$ can be chosen for optimal data rate.

\subsection{System Design Insights}
{
\begin{itemize}
        \item In cases where the deployment environment is apriori known to be stationary or for applications such as backhaul/midhaul connectivity, fixed wireless access, etc., the operator can provision the prescribed novel beam grouping strategy \texttt{CBE} for initial beam acquisition.
        \item For employing \texttt{CBE}, the operator must equip the RF front end of the gNodeB with advanced multi-antenna systems such as \cite{foo2015orthogonal} so as to form concurrent orthogonal beams with limited sidelobe radiation. Thus, there is a tradeoff between the initial access delay and the implementation cost and complexity that the operator has to optimize.
        \item Albeit the operator can configure the \ac{SSB} burst period in the gNodeB and the measurement gap in the user equipment, in the case of a highly dynamic environment, the best beam identity may change during an \ac{SSB} burst. For such scenarios, popular initial access mechanisms such as exhaustive and iterative search perform suboptimally. The operator must provision for such suboptimality in deriving the performance guarantees of the offered application.
        \item If the operator has apriori knowledge about the dynamics of the environment, it can determine the maximum suboptimal beam index that can become optimal during an \ac{SSB} burst. In such scenarios, the operator can employ the prescribed \texttt{K-SHES} algorithm that outperforms the classical and state-of-the-art initial access schemes.
        \item Finally, there exist optimal pairs of beam-dictionary size and initial access pilot fraction that maximize the downlink throughput of the users. For a given antenna array size, the operator must employ an optimal pilot fraction size as prescribed in our work to not only minimize the beam-selection error but also to provide sufficient resources in the communication phase.
    \end{itemize}
    }
\section{Conclusion}
\label{sec:C}
For the stationary environment, our proposed beam selection scheme \texttt{CBE} outperforms the state-of-the-art bandit algorithms in terms of the probability of error. For the non-stationary environment, we showed that the popular \texttt{SH} algorithm does not achieve an exponential error bound. For a know range of index of change, we proposed \texttt{K-SHES} that achieves an exponential bound and thus, can be employed in beam selection procedures where the state of the beams change during initial access. We employed \texttt{K-SHES} in a tandem beam refinement and data transmission scheme and highlighted key system design insights in terms of selection of beam codebook and partitioning of temporal resources. A detailed analysis of the type of {\it allowable} change distributions as well as handling multiple changes are indeed interesting directions of research and we are currently investigating the same. This will be reported in future work.

\appendices

\section{Proof of Lemma~\ref{cor:pm_bound}}
\label{app:CBE_miss}
  Let $\zeta_1 = \frac{a_1}{b_1}$. Following \eqref{eq:lim_a_b_1} we have $\zeta_1 > 1$ for diminishing $g$ and accordingly    
\begin{align}
    p_{\rm m} &= 1 - \mathcal{Q}_{\frac{T_{B_k}}{2}}\left(a_1, b_1\right) \nonumber \\
    &\overset{(a)}{\leq} \exp\left(-\frac{1}{2}\left(a_1^2 + b_1^2\right)\right)\sqrt{\mathcal{I}_0\left(2a_1b_1\right)}\sqrt{\frac{\zeta_1^{2\left(1 - \frac{T_{B_k}}{2}\right)}}{2\left(\zeta_1^2 - 1\right)}} \nonumber \\
    &\overset{(b)}{\leq} \exp\left(-\frac{1}{2}\left(a_1^2 + b_1^2\right)\right)\exp\left(a_1b_1\right)\sqrt{\frac{\zeta_1^{2\left(1 - \frac{T_{B_k}}{2}\right)}}{2\left(\zeta_1^2 - 1\right)}}  \nonumber \\
    & {\leq} \exp\left(-\frac{a_1^2}{2}\right)\exp\left(a_1b_1\right)\sqrt{\frac{\zeta_1^{2\left(1 - \frac{T_{B_k}}{2}\right)}}{2\left(\zeta_1^2 - 1\right)}} \nonumber \\
    &= C_{1}\exp\left(-\frac{g'}{2\sigma^2} \frac{T}{N\log  N}\right) \overset{(c)}{\leq} C_{1}\exp\left(-\frac{GT}{2N\sigma^2\log  N}\right). \nonumber 
\end{align}
$\mathcal{I}_0\left(\cdot\right)$ is the modified Bessel function of the first kind with order $0$. Step $(a)$ follows from the Cauchy-Schwarz inequality for the Marcum-Q function~\cite{annamalai2001cauchy}. The step $(b)$ follows from the following~\cite{luke1972inequalities}
\begin{align}
I_\nu(x) &< \frac{\cosh x}{\Gamma(\nu+1)\left(\frac{2}{x}\right)^\nu }   \leq  \frac{x^\nu}{2^\nu \nu!} \exp(x)
\implies \mathcal{I}_0(x) \leq   \exp(x).
    \nonumber 
\end{align}
The step (c) follows from the definition of $g'$.
Now consider
    $C_1 = \exp\left(a_1b_1\right) \sqrt{\frac{\zeta_1^{2\left(1 - \frac{T_{B_k}}{2}\right)}}{2\left(\zeta_1^2 - 1\right)}}.$
Due to \eqref{eq:lim_a_b_1} we have $
\lim_{g \to 0}\zeta_1 = \infty$ and thus,
 $\lim_{\zeta_1 \to \infty}    \sqrt{\frac{\zeta_1^{2\left(1 - \frac{T_{B_k}}{2}\right)}}{2\left(\zeta_1^2 - 1\right)}} = 0, \quad \forall T,$
and  $\lim_{g \to 0} \exp\left(a_1b_1\right) = 1$. Thus, from the limit rule of product, we have $\lim_{g \to 0} C_1 = 0$.

\section{Proof of Lemma~\ref{cor:pf_bound}}
\label{app:pf_bound}
The proof follows by considering $\zeta_0 = \frac{a_0}{b_0} < 1$ and applying the corresponding Cauchy-Schwarz bound for $Q_{\frac{T_{B_k}}{2}}(a_0,b_0)$ -
\begin{align}
    p_{\rm f} &= \mathcal{Q}_{\frac{T_{B_k}}{2}}\left(a_0, b_0\right) \nonumber \\
    &\overset{(a)}{\leq} \exp\left(-\frac{1}{2}\left(a_0^2 + b_0^2\right)\right)\sqrt{\mathcal{I}_0\left(2a_0b_0\right)}\sqrt{\frac{\zeta_0^{2\left(1 - \frac{T_{B_k}}{2}\right)}}{2\left(\zeta_0^2 - 1\right)}} \nonumber \\
    &{\leq} \exp\left(-\frac{1}{2}\left(a_0^2 + b_0^2\right)\right)\exp\left(a_0b_0\right)\sqrt{\frac{\zeta_0^{2\left(1 - \frac{T_{B_k}}{2}\right)}}{2\left(\zeta_0^2 - 1\right)}}  \nonumber \\
    & {\leq} \exp\left(-\frac{b_0^2}{2}\right)\exp\left(a_0b_0\right)\sqrt{\frac{\zeta_0^{2\left(1 - \frac{T_{B_k}}{2}\right)}}{2\left(\zeta_0^2 - 1\right)}} \nonumber \\
    &= C_{0}\exp\left(-\frac{GT}{2N\sigma^2\log  N}\right). \nonumber 
\end{align}
Unlike $p_{\rm m}$, the step $(a)$ follows since $\zeta_0 < 0$. Now consider
\begin{align}
    C_0 = \exp\left(a_0b_0\right) \exp\left(-N\right) \sqrt{\frac{\zeta_0^{2\left(1 - \frac{T_{B_k}}{2}\right)}}{2\left(\zeta_0^2 - 1\right)}}.
\end{align}
Due to \eqref{eq:lim_a_b_0} we have $
\lim_{g \to 0}\zeta_0 = 0$ and accordingly,
\begin{align}
 \lim_{\zeta_0 \to 0}    \sqrt{\frac{\zeta_0^{2\left(1 - \frac{T_{B_k}}{2}\right)}}{2\left(\zeta_0^2 - 1\right)}} = 0, \quad \forall T,
\end{align}
and  $\lim_{g \to 0} \exp\left(a_0b_0\right) = 1$. Thus, from the limit rule of product, we have $\lim_{g \to 0} C_1 = 0$.

\section{Proof of Lemma~\ref{lem:rc_bound}}
\label{app:rc_bound}
We have from the definition of $p_{i,j}(r_{\rm c})$,
\begin{align}
    &p_{i,j}(r_{\rm c}) = \mathbb{P}\left(\hat{\mu}_{i}(r_{\rm c})  - \hat{\mu}_j(r_{\rm c}) > 0 \mid {\bf f}_i, {\bf f}_j \in S_{r_{\rm c}}\right) \nonumber \\
    &=\mathbb{P}\left(\frac{1}{n_{r_{\rm c}}} \sum_{k= 1}^{n_{r_{\rm c}}} R_{i}\left(k\right)  - \frac{1}{n_{r_{\rm c}}} \left[\sum_{l = 1}^{n' - 1} R_j^-\left(l\right) + \right.\right.\nonumber \\
    &\left.\left. \qquad \sum_{m = 1}^{n_{r_{\rm c}}-n' + 1} R_j^+\left(m\right)\right] > 0  \mid {\bf f}_i, {\bf f}_j \in S_{r_{\rm c}}\right) \nonumber \\
    &=  \mathbb{P}\left(\frac{n' - 1}{n_{r_{\rm c}}}\left[\hat{\mu}_{i, n' - 1} - \hat{\mu}_{j, n' - 1}^-\right] + \frac{n_{r_{\rm c}} - n' + 1}{n_{r_{\rm c}}} \cdot \right. \nonumber \\
    & \left. \qquad \left[\hat{\mu}_{i, n_{r_{\rm c}} - n' + 1} - \hat{\mu}_{j, n_{r_{\rm c}} - n' + 1}^+\right] > 0  \mid {\bf f}_i, {\bf f}_j \in S_{r_{\rm c}}\right) \nonumber \\
    &=  \mathbb{P}\left(\left[\hat{\mu}_{i, n_{r_{\rm c}}} - \mu_{i}\right] + \frac{n' - 1}{n_{r_{\rm c}}}\left[\mu_j^- - \hat{\mu}_{j, n' - 1}^-\right] +\right. \nonumber \\
    &\quad \left.   \frac{n_{r_{\rm c}} - n' + 1}{n_{r_{\rm c}}}\left[\mu_j^+ - \hat{\mu}^+_{j, n_{r_{\rm c}}-n' + 1}\right] + \frac{n' - 1}{n_{r_{\rm c}}} \Delta^-_{i,j} + \right. \nonumber \\
    &\left. \qquad \frac{n_{r_{\rm c}} - n' + 1}{n_{r_{\rm c}}}\Delta^+_{i,j} > 0  \mid {\bf f}_i, {\bf f}_j \in S_{r_{\rm c}}\right). \label{eq:expression'ondition}
\end{align}
Now, let us note that
\begin{align}
    Z &= \left[\hat{\mu}_{i, n_{r_{\rm c}}} - \mu_{i}\right] + \frac{n' - 1}{n_{r_{\rm c}}}\left[\mu_j^- - \hat{\mu}_{j, n' - 1}^-\right] +  \nonumber \\
    &\frac{n_{r_{\rm c}} - n' + 1}{n_{r_{\rm c}}}\left[\mu_j^+ - \hat{\mu}^+_{j, n_{r_{\rm c}}-n' + 1}\right] \sim \mathcal{N}\left(0, \sigma'^2_{ij}\right). \nonumber 
\end{align}
Accordingly,
\begin{align}
    p_{i,j}(r_{\rm c}) &\leq \frac{1}{{\sum_{n' = 0}^{n_{r_{\rm c}}} q_{n}(n')}}\sum_0^{n^*_i}\exp\left(-\frac{\Delta'^2_{i,j}}{2\sigma'^2_{ij}}\right) q_{n}(n') + \sum_{n^*_i}^{n_{r_{\rm c}}} q_{n}(n') \nonumber \\
    &\leq  \exp\left(-\frac{\Delta^2_{\min}}{2\sigma'^2_{ij}}\right)  F_{n'}\left(n_i^* \right) +1 - {F}_{n'}\left(n_i^* \right) \nonumber \\
    &\leq \exp\left(-\frac{\Delta^2_{\min}T}{4N \log N\sigma^2\mu_{\max}}\right)  F_{n'}\left(n_i^* \right) +1 - {F}_{n'}\left(n_i^* \right). \nonumber 
\end{align}
where in step (a), $\Delta'_{i,j} = \frac{n' - 1}{n_{r_{\rm c}}} \Delta^-_{i,j} +  \frac{n_{r_{\rm c}} - n' + 1}{n_{r_{\rm c}}}\Delta^+_{i,j}$, $\sigma'^2_{ij} = \frac{\sigma_i^2}{n_{r_c}} + \frac{n'{\sigma_j^-}^2}{n^2_{r_c}} + \frac{(n_{r_c} - n'){\sigma_j^+}^2}{n^2_{r_c}}$.

\section{Proof of Lemma~\ref{lem:ec}}
\label{sec:proof_ec}
The analysis considers three phases - i) rounds before $r_{{\rm c}}$, ii) the $r_{{\rm c}}-$th round, and iii) the rounds after $r_{{\rm c}}$.

For the rounds before $r_{\rm c}$, let $N'_r$ be the number of arms from the bottom ${|\mathcal{S}_r|} - K$ arms that have the estimate of their means larger than the estimate of the $K-$th arm. We have $\forall r < r_{{\rm C}}$
\begin{align}
    &\mathbb{E}[N'_r \mid r < r_{\rm c} \leq r^*] \leq \sum_{q = K+1}^{|\mathcal{S}_r|} \exp\left(- \frac{1}{2}\Delta_{Kq}^2 n_r\right) \nonumber \\
    & \leq (|\mathcal{S}_r| - K - 1) \exp\left(- \frac{1}{2} \Delta_{\min}^2 \frac{T}{N\log N}\right) \nonumber.
\end{align}
Consequently, the probability that the $K-$th arm is eliminated in the $r-$th round ($r \leq r_{\rm c} \leq r^*$) is upper bounded by
\begin{equation} 
\begin{split}
    &\mathbb{P}\left(N'_r > \frac{|S_r|}{2}  \mid  r \leq r_{\rm c} \leq r^*\right) \nonumber \\
    &\leq \frac{2}{|S_r|} \left[(|\mathcal{S}_r| - K - 1)\exp\left(-\frac{1}{2} \Delta_{\min}^2 \frac{T}{N \log N}\right)\right] \nonumber \\
     &\leq 2\exp\left(-\frac{1}{2} \Delta_{\min}^2 \frac{T}{N \log N}\right). \nonumber
\end{split}
\end{equation} 
Accordingly, the probability that the $K-$th arm is eliminated until the round $r_{{\rm c}}$ is upper bounded as
\begin{align}
    &p_{{\rm e}_1}(r_{{\rm c}}) = \mathbb{P}\left(\mathcal{E}_K([r_{\rm c - 1}]) \mid r_{\rm c} \leq r^*\right)  \nonumber \\
    &\leq \mathbb{E}_{r_{\rm c} \leq r^*}\left[r_{{\rm c}} - 1\right] \left[2\exp\left(-\frac{1}{2} \Delta_{\min}^2 \frac{T}{N \log N}\right)\right]. \nonumber 
\end{align}

Next, if the $K-$th arm survives until the $r_{{\rm c}}-$th round, the probability that it is eliminated in the $r_{{\rm c}}-$th round is evaluated as (following Lemma~\ref{lem:changeround})
\begin{align}
    p_{{\rm e}_2}(r_{{\rm c}}) &=  \mathbb{P}\left(\mathcal{E}_K(r) \mid \mathds{1}\left(\mathcal{E}_K\left([r_{\rm c} - 1]\right)\right) = 0 , r \leq r_{\rm c} \leq r^*\right) \nonumber \\
    &\leq  2 \mathbb{E}_{r_{\rm c} \leq r^*}\left[1 - 
    F_{n'}\left(n_{\max} \right) \left(1 -  \right.\right. \nonumber \\
    &\qquad \left.\left.\exp\left(- \frac{T\Delta^2_{\min}}{2N\log N\sigma^2_{\max}}\right)\right)\right] \nonumber \\
    &\overset{(a)}{\leq} K \exp\left(- \frac{T\Delta^2_{\min}}{2N\log N\sigma^2_{\max}}\right). \nonumber
\end{align}
Step $(a)$ is because before $r^*$ only 1 arm among the set of beams inferior to the $K-$th beams needs a higher estimate than the $K-$th arm for it to be eliminated. Finally, given that the $K-$th arm has survived until the end of $r_{{\rm c}}$ the probability it is eliminated at the end of the play is upper bounded as
\begin{align}
  p_{{\rm e}_3}(r_{{\rm c}}) &=  \mathbb{P}\left(\mathcal{E}_K([r_{\rm c}, \log N]) \mid \mathds{1}\left(e_K\left(r^*\right)\right) = 0 , r \leq r_{\rm c} \leq r^*\right)  \nonumber \\ 
  &\leq  \mathbb{E}_{r_{\rm c} \leq r^*}\left[\log N - r_{{\rm c}_1}\right] 2\exp\left(-\frac{1}{2} \Delta_{\min}^2 \frac{T}{N \log N}\right). \nonumber 
\end{align}
Thus, the total probability that the arm $K$ is eliminated given that the change occurs in the first $\log \frac{N}{2K}$ rounds is upper bounded using the union bound as
\begin{align}
  &\mathbb{P}\left(\mathcal{E}_K([\log  N]) \mid r_{\rm c} \leq r^* \right) \nonumber \\
  & \leq \mathbb{E}\left[p_{{\rm e}_1}(r_{{\rm c}}) + p_{{\rm e}_2}(r_{{\rm c}}) + p_{{\rm e}_3}(r_{{\rm c}}) \mid r_{\rm c} \leq r^*\right] \nonumber \\
  & \leq \mathbb{E}_{r_{\rm c}, n'} \left[ \left(r_{\rm c} - 1 \right)  \left[2\exp\left(-\frac{1}{2}\frac{\Delta^2_{\min}T}{N\log  N}\right) \right] \right.+ \nonumber \\
    &\left.  \left[K\exp\left(-\frac{1}{2\sigma^2_{\max}}\frac{\Delta^2_{\min}T}{N\log N }\right)\right] + \right.\nonumber \\
    &\left.2\left(\log  N - r_{\rm c}\right) \left[\exp\left(-\frac{1}{2} \frac{\Delta^2_{\min} T}{N \log  N}\right)\right] \mid r_{\rm c} \leq r^* \right]. \nonumber
\end{align}

\section{Proof of Lemma~\ref{lem:lc}}
\label{sec:proof_lc}
We have $\forall r < r^*$
\begin{align}
    &\mathbb{E}[N'_r \mid r \leq r^* \leq r_{\rm c}] \leq \sum_{q = K+1}^{|\mathcal{S}_r|} \exp\left(- \frac{1}{2}\Delta_{Kq}^2 n_r\right) \nonumber \\
    & \leq (|\mathcal{S}_r| - K - 1) \exp\left(- \frac{1}{2} \Delta_{\min}^2 \frac{T}{N\log N}\right) \nonumber.
\end{align}
Consequently, the probability that the $K-$th arm is eliminated in the $r-$th round ($r \leq r^* \leq r_{\rm c}$) is upper bounded by
\begin{equation} 
\begin{split}
    &\mathbb{P}\left(N'_r > \frac{|S_r|}{2}  \mid  r \leq r_{\rm c} \leq r^*\right) \nonumber \\
    &\leq \frac{2}{|S_r|} \left[(|\mathcal{S}_r| - K - 1)\exp\left(-\frac{1}{2} \Delta_{\min}^2 \frac{T}{N \log N}\right)\right] \nonumber \\
     &\leq 2\exp\left(-\frac{1}{2} \Delta_{\min}^2 \frac{T}{N \log N}\right). \nonumber
\end{split}
\end{equation} 
Accordingly, the probability that the $K-$th arm is eliminated until the round $r_{*}$ is upper bounded as
\begin{align}
    &p_{{\rm l}_1}(r_{{\rm c}}) = \mathbb{P}\left(\mathcal{E}_K([r^*]) \mid r_{\rm c} > r^*\right)  \nonumber \\
    &\leq  2r^*\exp\left(-\frac{1}{2} \Delta_{\min}^2 \frac{T}{N \log N}\right).\nonumber 
\end{align}
Given that the $K-$th arm has survived until the end of $r_{{\rm c}_2}$ the probability it is eliminated at the end of the play is upper bounded as
\begin{align}
    p_{{\rm l}_2}(r_{{\rm c}_2}) &= 2\left(\log N - r_{\rm c}\right) \exp\left(-\frac{1}{2} \Delta_{\min}^2 \frac{T}{N \log N}\right).
\end{align}
Thus, the total probability that the arm $K$ is eliminated given that the change occurs in the first $\log \frac{N}{2K}$ rounds is given by
\begin{align}
    p_{\rm l} &\leq \mathcal{T}_2(r_{\rm c}) + 2\log N \exp\left(-\frac{1}{2} \Delta_{\min}^2 \frac{T}{N \log N}\right).
\end{align}

\bibliography{references.bib}
\bibliographystyle{IEEEtran}

\end{document}